\documentclass[11pt]{article}
\usepackage[ruled,vlined,linesnumbered]{algorithm2e}
\usepackage{amsfonts}
\usepackage{amssymb}
\usepackage{amstext}
\usepackage{amsmath}
\usepackage{amsthm}
\usepackage{authblk}
\usepackage{bbm}
\usepackage{enumitem}
\usepackage{geometry}
\usepackage{mathrsfs}
\usepackage{tikz}
\usepackage{float}
\usepackage{accents}
\usepackage[backref=page]{hyperref}
\usepackage{cleveref}
\usepackage{caption}
\usepackage{microtype}
\usepackage[colorinlistoftodos,prependcaption]{todonotes}
\usetikzlibrary{arrows.meta}

\newtheorem{theorem}{Theorem}[section]
\newtheorem{lemma}[theorem]{Lemma}
\newtheorem{corollary}[theorem]{Corollary}

\newtheorem{claim}[theorem]{Claim}
\theoremstyle{definition}
\newtheorem{definition}[theorem]{Definition}
\newtheorem{example}[theorem]{Example}

\newcommand{\R}{\mathbb R}
\newcommand{\Z}{\mathbb Z}
\newcommand{\N}{\mathbb N}

\newcommand{\1}{\mathbbm 1}
\newcommand{\abs}[1]{\ensuremath{\left|#1\right|}}
\newcommand{\size}[1]{\ensuremath{\left|#1\right|}}
\newcommand{\norm}[1]{\ensuremath{\left\|#1\right\|}}

\newcommand{\ole}[1]{\accentset{\leftarrow}{#1}}
\newcommand{\set}[1]{\ensuremath{\left\{#1\right\}}}
\newcommand{\pr}[1]{\ensuremath{\left(#1\right)}}

\DeclareMathOperator*{\argmax}{arg\,max}
\DeclareMathOperator*{\argmin}{arg\,min}
\DeclareMathOperator{\subto}{s.t.}

\DeclareMathOperator{\sgn}{sgn}

\DeclareMathOperator{\conv}{conv}
\DeclareMathOperator{\dom}{dom}

\DeclareMathOperator{\pred}{pred}

\tikzstyle{vertex}=[circle,draw,align=center]
\tikzstyle{node}=[circle,draw,fill=white!100,inner sep=0pt,minimum size=0.15cm,align=center]
\tikzstyle{weight} = [font=\small, black]
\tikzstyle{edge} = [draw,-]
\tikzstyle{selected edge} = [draw,line width=3pt,-,gray!70]
\tikzstyle{matched edge} = [draw,line width=3pt,-]
\tikzstyle{dashed edge} = [draw,dashed]
\tikzstyle{selected vertex} = [node, fill=gray!70, text=white]

\tikzstyle{remove selected vertex} = [vertex, fill=white!100]
\tikzstyle{blue colored edge} = [draw,line width=3pt,-,blue!80]
\tikzstyle{blue colored curved edge left} = [draw,line width=3pt,-,blue!80,bend left]
\tikzstyle{blue colored curved edge right} = [draw,line width=3pt,-,blue!80,bend right]
\tikzstyle{green colored edge} = [draw,line width=3pt,-,green!80]
\tikzstyle{remove selected edge} = [draw,line width=4pt,-,white!100]
\tikzstyle{remove matched edge} = [draw,line width=4pt,-,white!80,bend right]

\SetArgSty{textup}
\DontPrintSemicolon
\SetKw{Continue}{continue}
\SetKwInOut{Input}{input}
\SetKwInOut{Output}{output}

\usepackage[bibliography=common]{apxproof}
\newtheoremrep{apxtheorem}[theorem]{Theorem}
\newtheoremrep{apxlemma}[theorem]{Lemma}

\hypersetup{
    colorlinks,
    linkcolor={red!80!black},
    citecolor={blue!50!black},
    urlcolor={blue!80!black}
}

\usepackage{setspace}
\newcounter{mycomment}

\title{An Accelerated Newton--Dinkelbach Method and its Application to Two Variables Per Inequality Systems\thanks{This project has received funding from the European Research Council (ERC) under the European Union's Horizon 2020 research and innovation programme (grant agreement nos.~757481--ScaleOpt and 805241--QIP).}}

\author[1]{Daniel Dadush}
\author[2]{Zhuan Khye Koh} 
\author[2]{Bento Natura}
\author[2]{L\'{a}szl\'{o} A. V\'{e}gh}

\affil[1]{Centrum Wiskunde \& Informatica, Netherlands.}
\affil[2]{Department of Mathematics, London School of Economics, UK.}

\date{\tt{dadush@cwi.nl,$\{$z.koh3,b.natura,l.vegh$\}$@lse.ac.uk}}

\begin{document}

\maketitle

\begin{abstract}
We present an accelerated, or `look-ahead' version of the Newton--Dinkelbach method, a well-known technique for solving fractional and parametric optimization problems. 
This acceleration halves the Bregman divergence between the current iterate and the optimal solution within every two iterations. 
Using the Bregman divergence as a potential in conjunction with combinatorial arguments, we obtain strongly polynomial algorithms in three applications domains:
{\em (i)} For linear fractional combinatorial optimization, we show a convergence bound of $O(m\log m)$ iterations; the previous best bound was $O(m^2\log m)$ by Wang et al.~(2006).
{\em (ii)} We obtain a strongly polynomial label-correcting algorithm for solving  linear feasibility systems with two variables per inequality (2VPI).
For a 2VPI system with $n$ variables and $m$ constraints, our algorithm runs in $O(mn)$ iterations.
Every iteration takes $O(mn)$ time for general 2VPI systems, and 
$O(m + n\log n)$ time for the special case of deterministic Markov Decision Processes (DMDPs).
This extends and strengthens a previous result by Madani~(2002) that showed a weakly polynomial bound for a variant of the Newton–Dinkelbach
method for solving DMDPs.
{\em (iii)} We give a simplified variant of the parametric submodular function minimization result by Goemans et al.~(2017).
\end{abstract}

\section{Introduction}
Linear fractional optimization problems are well-studied in combinatorial optimization. Given a closed domain $\mathcal{D}\subseteq \R^m$ and $c,d\in \R^m$ such that $d^\top x> 0$ for all $x\in\mathcal{D}$, the problem is
\begin{equation}\label{prob:minratio-D}
\inf c^\top x/d^\top x\quad \mbox{s.t. }x\in\mathcal{D}\, .
\end{equation}
The domain $\mathcal{D}$ could be either a convex set or a discrete set $\mathcal{D}\subseteq \{0,1\}^m$. Classical examples include finding minimum cost-to-time ratio cycles and minimum ratio spanning trees. 
One can equivalently formulate \eqref{prob:minratio-D} as a parametric search problem. Let
\begin{equation}\label{prob:f-D}
f(\delta)=\inf\{(c - \delta d)^\top x:\, x\in\mathcal{D}\}\, ,
\end{equation}
be a concave and decreasing function. Assuming \eqref{prob:minratio-D} has a finite optimum $\delta$, it corresponds to the unique root $f(\delta)=0$.

A natural question is to investigate how the computational complexity of solving the minimum ratio problem \eqref{prob:minratio-D} may depend on the complexity of the corresponding linear optimization problem $\min c^\top x$ s.t.~$x\in\mathcal{D}$. Using the reformulation \eqref{prob:f-D}, one can reduce the fractional problem to the linear problem via binary search; however, the number of iterations needed to find an exact solution may depend on the bit complexity of the  input. A particularly interesting question is: assuming there exists a strongly polynomial algorithm for linear optimization over a domain $\mathcal{D}$, can we find a strongly polynomial algorithm for linear fractional optimization over the same domain?

A seminal paper by Megiddo \cite{journals/mor/Megiddo79} introduced the \emph{parametric search} technique to solve linear fractional combinatorial optimization problems. 
He showed that if the linear optimization algorithm only uses $p(m)$ additions and $q(m)$ comparisons, then there exists an $O(p(m)(p(m)+q(m))$ algorithm for the linear fractional optimization problem. This in particular yielded the first strongly polynomial algorithm for the minimum cost-to-time ratio cycle problem.
On a very high level, parametric search works by simulating the linear optimization algorithm for the parametric problem \eqref{prob:f-D}, with the parameter $\delta\in \R$ being indeterminate. 

A natural alternative approach is to solve \eqref{prob:f-D} using a standard root finding algorithm. Radzik \cite{conf/focs/Radzik92} showed that for a discrete domain $\mathcal{D}\subseteq \{0,1\}^m$, the \emph{discrete Newton} method---in this context,
also known as \emph{Dinkelbach's method}~\cite{journals/mansci/Dinkelbach67}---terminates in a strongly polynomial number of iterations. In contrast to parametric search, there are no restrictions on the possible operations in the linear optimization algorithm. In certain settings, such as the maximum ratio cut problem, the discrete Newton method outperforms parametric search; we refer to the comprehensive survey by Radzik \cite{book/hco/Radzik98} for details and comparison of the two methods.

\subsection{Our contributions}
We introduce a new, \emph{accelerated variant of Newton's method for univariate functions}. Let $f:\R\to\R\cup\set{-\infty}$ be a concave function. Under some mild assumptions on $f$, our goal is to either find the largest root, or show that no root exists. Let $\delta^*$ denote the largest root, or in case $f<0$, let $\delta^*$ denote the largest maximizer of $f$.
For simplicity, we now describe the method for differentiable functions.
This will not hold in general:  functions of the form \eqref{prob:f-D} will be piecewise linear if $\mathcal{D}$ is finite or polyhedral. The algorithm description in
Section~\ref{sec:newton} uses a form with supergradients (that can be choosen arbitrarily between the left and right derivatives).

 The standard Newton method, also used by Radzik, proceeds through iterates $\delta^{(1)}>\delta^{(2)}>\ldots>\delta^{(t)}$ such that $f(\delta^{(i)})\le 0$, and updates $\delta^{(i+1)}=\delta^{(i)}-f(\delta^{(i)})/f'(\delta^{(i)})$.

Our new variant uses a more aggressive \emph{`look-ahead'} technique. At each iteration, we compute $\delta= \delta^{(i)}-f(\delta^{(i)})/ f'(\delta^{(i)})$, and jump ahead to $\delta'=2\delta-\delta^{(i)}$. In case $f(\delta')\le 0$ and $f'(\delta')<0$, we update $\delta^{(i+1)}=\delta'$; otherwise, we continue with the standard iterate $\delta$.

This modification leads to an improved and at the same time simplified analysis based on the \emph{Bregman divergence} $D_f(\delta^*,\delta^{(i)})=f(\delta^{(i)})-f(\delta^*)+ f'(\delta^{(i)}) (\delta^*-\delta^{(i)})$. We show that this decreases by a factor of two between any two iterations. 

A salient feature of the algorithm is that it handles both feasible and infeasible outcomes in a unified framework. 
In the context of linear fractional optimization, this means that the assumption $d^{\top}x>0$ for all $x\in \mathcal{D}$ in \eqref{prob:minratio-D} can be waived.
Instead, $d^{\top}x>0$ is now added as a feasibility constraint to \eqref{prob:minratio-D}.
This generalization  is important when we use the algorithm to solve two variables per inequality systems.

\medskip

This general result leads to improvements and simplifications of a number of algorithms using the discrete Newton method.
\begin{itemize}
	\item For \emph{linear fractional combinatorial optimization}, namely the setting \eqref{prob:minratio-D} with $\mathcal{D}\subseteq\{0,1\}^m$, we obtain an
	$O(m\log m)$ bound on the number of iterations, a factor $m$ improvement over the previous best bound $O(m^2\log m)$ by Wang et al.~\cite{journals/jgo/WangYZ06} from 2006. We remark that Radzik's first analysis \cite{conf/focs/Radzik92} yielded a bound of $O(m^4\log^2 m)$ iterations, improved to $O(m^2 \log^2 m)$ in \cite{book/hco/Radzik98}.
	\item Goemans et al.~\cite{conf/ipco/GoemansGJ17} used the discrete Newton method 
to obtain a strongly polynomial algorithm for parametric submodular function minimization. We give a simple new variant of this result with the same asymptotic running time, using the accelerated algorithm.
	\item For \emph{two variable per inequality (2VPI)} systems, we obtain a \emph{strongly polynomial label-correcting algorithm}. This will be discussed in more detail next.
\end{itemize}

\subsection{Two variables per inequality systems}\label{sec:tvpi-intro}
A major open question in the theory of linear programming (LP) is whether there exists a strongly polynomial algorithm for LP.
This problem is one of Smale's eighteen mathematical challenges for the twenty-first century \cite{journals/mi/Smale98}.
An LP algorithm is \emph{strongly polynomial} if it only uses  elementary arithmetic operations ($+,-,\times,/$) and comparisons, and the number of such operations  is polynomially bounded in the number of variables and constraints.
Furthermore, the algorithm needs to be in PSPACE, i.e.~the numbers occurring in the computations must remain polynomially bounded in the input size.

The notion of a strongly polynomial algorithm was formally introduced by Megiddo \cite{journals/siamcomp/Megiddo83} in 1983 (using the term \emph{`genuinely polynomial'}), where he gave the first such algorithm for \emph{two variables per inequality (2VPI)} systems. These are feasibility LPs where every inequality contains at most two variables. More formally, let $\mathcal{M}_2(n,m)$ be the set of $n\times m$ matrices with at most two nonzero entries per column. A 2VPI system is of the form  $A^{\top}y\leq c$ for $A\in \mathcal{M}_2(n,m)$ and $c\in \R^m$.

If we further require that every inequality has at most one positive and at most one negative coefficient, it is called a \emph{monotone two variables per inequality} (M2VPI) system. 
A simple and efficient reduction is known from 2VPI systems with $n$ variables and $m$ inequalities to M2VPI systems with $2n$ variables and $\le 2m$ inequalities \cite{journals/tcs/EdelsbrunnerRW89,journals/mp/HochbaumMNT93} (sketch in Appendix \ref{subsec: 2VPI_reduction}).

\paragraph{Connection between 2VPI and parametric optimization}
An M2VPI system has a natural graphical interpretation: after normalization, we can assume every constraint is of the form $y_u-\gamma_e y_v\le c_e$.
Such a constraint naturally maps to an arc $e=(u,v)$ with \emph{gain factor} $\gamma_e$ and cost $c_e$. 
Based on Shostak's work \cite{journals/jacm/Shostak81} that characterized feasibility in terms of this graph, Aspvall and Shiloach \cite{journals/siamcomp/AspvallS80} gave the 
first weakly polynomial algorithm for M2VPI systems.

We say that a directed cycle $C$ is \emph{flow absorbing} if $\prod_{e\in C}\gamma_e<1$ and \emph{flow generating} if $\prod_{e\in C}\gamma_e>1$. Every flow absorbing cycle $C$ implies an upper bound for every variable $y_u$ incident to $C$; similarly, flow generating cycles imply lower bounds. The crux of Aspvall and Shiloach's algorithm is to find the tightest upper and lower bounds for each variable $y_u$. 

Finding these bounds corresponds to solving fractional optimization problems of the form \eqref{prob:minratio-D}, where $\mathcal{D}\subseteq \R^m$ describes `generalized flows' around cycles. The paper \cite{journals/siamcomp/AspvallS80} introduced the \emph{Grapevine} algorithm---a natural modification the Bellman-Ford algorithm---to decide whether the optimum ratio is smaller or larger than a fixed value $\delta$. The optimum value can found using binary search on the parameter.

Megiddo's strongly polynomial algorithm \cite{journals/siamcomp/Megiddo83}  replaced the binary search framework in Aspvall and Shiloach's algorithm by extending the parametric search technique in \cite{journals/mor/Megiddo79}.
Subsequently, Cohen and Megiddo \cite{journals/siamcomp/CohenM94} devised faster strongly polynomial algorithms for the problem.
The current fastest strongly polynomial algorithm is given by Hochbaum and Naor \cite{journals/siamcomp/HochbaumN94}, an efficient Fourier--Motzkin elimination with running time of $O(mn^2\log m)$. 

\paragraph{2VPI via Newton's method} Since Newton's method proved to be an efficient and viable alternative to parametric search, a natural question is to see whether it can solve the parametric problems occuring in 2VPI systems.  Radzik's fractional combinatorial optimization results \cite{conf/focs/Radzik92,book/hco/Radzik98} are not directly applicable, since the domain $\mathcal{D}$ in this setting is a polyhedron and not a discrete set.\footnote{The problem could be alternatively formulated with $\mathcal{D}\subseteq \{0,1\}^m$ but with nonlinear functions instead of $c^\top x$ and $d^\top x$.} Madani \cite{conf/aaai/Madani02} used a variant of the Newton--Dinkelbach method as a tool to analyze the convergence of policy iteration on \emph{deterministic Markov Decision Processes (DMDPs)}, a special class of M2VPI systems (discussed later in more detail). He obtained a weakly polynomial convergence bound; it remained open whether such an algorithm could be strongly polynomial.

\paragraph{Our 2VPI algorithm} 
We introduce a new type of strongly polynomial 2VPI algorithm by 
combining the accelerated Newton--Dinkelbach method with a \emph{`variable fixing'} analysis.
Variable fixing was first introduced in the seminal work of Tardos \cite{journals/combinatorica/Tardos85} on minimum-cost flows, and has been a central idea of 
 strongly polynomial algorithms, see in particular \cite{journals/jacm/GoldbergT89,journals/algorithmica/RadzikG94} for cycle cancelling minimum-cost flow algorithms, and \cite{OV20,journals/mor/Vegh17} for maximum generalized flows, a dual to the 2VPI problem.
 
  We show that for every iterate $\delta^{(i)}$,  there is a constraint that has been `actively used' at $\delta^{(i)}$ but will not be used ever again after a strongly polynomial number of iterations. The analysis combines the decay in \emph{Bregman-divergence} shown in the general accelerated Newton--Dinkelbach analysis with a combinatorial \emph{`subpath monotonicity'} property. 

 Our overall algorithm can be seen as an extension of Madani's DMDP algorithm. In particular, we adapt his `unfreezing' idea: the variables $y_u$ are admitted to the system one-by-one, and the accelerated Newton--Dinkelbach method is used to find the best `cycle bound' attainable at the newly admitted $y_u$ in the graph induced by the current variable set.
This returns a feasible solution or reports infeasibility within $O(m)$ iterations.
As every iteration takes $O(mn)$ time, our overall algorithm terminates in $O(m^2n^2)$ time. 
For the special setting of deterministic MDPs, the runtime per iteration improves to $O(m + n\log n)$, giving a total runtime of $O(mn(m+n\log n))$.

Even though our running time bound is worse than the state-of-the-art 2VPI algorithm \cite{journals/siamcomp/HochbaumN94}, it is of a very different nature from all previous 2VPI algorithms.
 In fact, our algorithm is a \emph{label correcting algorithm}, naturally fitting to the family of algorithms used in other 
combinatorial optimization problems  
with constraint matrices from $\mathcal{M}_2(n,m)$ such as 
maximum flow, shortest paths, minimum-cost flow, and generalized flow problems. 
We next elaborate  on this connection.

\paragraph{Label-correcting algorithms} An important special case of M2VPI systems corresponds to the  shortest paths problem: given a directed graph $G=(V,E)$ with target node $t\in V$ and arc costs $c\in\R^E$, we associate constraints $y_u-y_v\le c_{e}$ for every arc $e=(u,v)\in E$ and $y_t=0$. If the system is feasible and bounded, the pointwise maximal solution 
 corresponds to the shortest path labels to $t$; an infeasible system contains a negative cost cycle. A generic label-correcting algorithm maintains distance labels $y$ that are upper bounds on the shortest path distances to $t$. The labels are decreased according to violated constraints. Namely, if  $y_u-y_v>c_{e}$, then decreasing $y_u$ to $c_{e}+y_v$ gives a smaller valid distance label at $u$. We terminate with the shortest path labels once all constraints are satisfied.
The Bellman--Ford algorithm for the  shortest paths problem is a particular implementation of the generic label-correcting algorithm; we refer the reader to \cite[Chapter 5]{books/amo} for more details.

It is a natural question if label-correcting algorithms can be extended to general M2VPI systems, where constraints are of the form $y_u-\gamma_{e} y_v\le c_{e}$ for  `gain/loss factors' $\gamma_{e} \in \R_{>0}$ associated with each arc. A fundamental property of M2VPI systems is that, whenever bounded, a unique pointwise maximal solution exists, i.e.~a feasible solution $y^*$ such that $y\le y^*$ for every feasible solution $y$.
A label-correcting algorithm for such a setting can be naturally defined as follows. Let us assume that the problem is bounded. %
The algorithm should proceed via a decreasing sequence $y^{(0)}\ge y^{(1)}\ge\ldots\ge y^{(k)}$ of labels that are all valid upper bounds on any feasible solution $y$ to the system. The algorithm either terminates with the unique pointwise maximal solution $y^{(k)}=y^*$, or finds an infeasibility certificate.

The basic label-correcting operation is the `arc update', decreasing $y_u$ to $\min\{y_u,c_{e}+\gamma_{e}y_v\}$ for some arc $e=(u,v)\in E$. Such updates suffice in the shortest path setting. However, in the general setting arc operations only  may not lead to  finite termination. Consider a system with only two variables, $y_u$ and $y_v$, and two constraints, $y_u-y_v\le 0$, and $y_v-\frac{1}{2}y_u\le -1$. The alternating sequence of arc updates  converges to  $(y^*_u,y^*_v)=(-2,-2)$, but does not finitely terminate. In this example, we can `detect' the cycle formed by the two arcs, that implies the bound $y_u-\frac{1}{2}y_u\le -1$.

Shostak's~\cite{journals/jacm/Shostak81} result demonstrates that arc updates, together with such `cycle updates' should be sufficient for finite termination. 
Our M2VPI algorithm amounts to the first strongly polynomial label-correcting algorithm for general M2VPI systems, using arc updates and cycle updates.

\paragraph{Deterministic Markov decision processes}
A well-studied special case of M2VPI systems in which $\gamma\leq \1$ is known as \emph{deterministic Markov decision process} (DMDP). A \emph{policy} corresponds to selecting an outgoing arc from every node, and the objective is to find a policy that minimizes the total discounted cost over an infinite time horizon.
 The pointwise maximal solution of this system corresponds to the optimal values of a policy.

The standard policy iteration, value iteration, and simplex algorithms can be all interpreted as variants of the label-correcting framework.\footnote{The value sequence may violate monotonicity in certain cases of value iteration. %
} 
Value iteration can be seen as a generalization of the Bellman--Ford algorithm to the DMDP setting.
As our previous example shows, value iteration may not be finite. 
One could still consider as the termination criterion the point where value iteration `reveals' the optimal policy, i.e.~updates are only performed using constraints that are tight in the optimal solution. 
If each discount factor $\gamma_{uv}$ is at most $\gamma'$ for some $\gamma'>0$, then it is well-known that value iteration converges at the rate $1/(1-\gamma')$. 
This is in fact true more generally, for nondeterministic MDPs. 
However, if the discount factors can be arbitrarily close to 1, then Feinberg and Huang \cite{journals/orl/FeinbergH14} showed that value iteration cannot reveal the optimal policy in strongly polynomial time even for DMDPs. 
Post and Ye \cite{journals/mor/PostY15} proved that simplex with the highest gain pivoting rule is strongly polynomial for DMDPs; this was later improved by Hansen et al.~\cite{conf/soda/HansenKZ14}.
These papers heavily relies on the assumption $\gamma\le \1$, and does not seem to extend to general M2VPI systems.

Madani's previously mentioned work \cite{conf/aaai/Madani02} used a variant of the Newton--Dinkelbach method as a tool to analyze the convergence of policy iteration on deterministic MDPs, and derived 
a weakly polynomial runtime bound.

\paragraph{Paper organization} 
We start by giving preliminaries and introducing notation in Section \ref{sec:preliminaries}.
In Section \ref{sec:newton}, we present an accelerated Newton's method for univariate concave functions, and apply it to linear fractional combinatorial optimization and linear fractional programming.
Section \ref{sec:2vpi} contains our main application of the method to the 2VPI problem.
Our results on parametric submodular function minimization are in Section \ref{sec:sfm}.
Missing proofs can be found in the appendix.

\section{Preliminaries}
\label{sec:preliminaries}

Let $\R_+$ and $\R_{++}$ be the nonnegative and positive reals respectively, and denote $\bar{\R} := \R\cup\set{\pm \infty}$. 
Given a proper concave function $f:\R\rightarrow \bar{\R}$, let $\dom(f):=\set{x:-\infty<f(x)<\infty}$ be the effective domain of $f$. 
For a point $x_0\in \dom(f)$, denote the set of supergradients of $f$ at $x_0$ as $\partial f(x_0):=\set{g:f(x) \leq f(x_0) + g(x-x_0) \;\forall x\in \R}$.
If $x_0$ is in the interior of $\dom(f)$, then $\partial f(x_0)=[f'_-(x_0),f'_+(x_0)]$, where $f'_-(x_0)$ and $f'_+(x_0)$ are the left and right derivatives.
Throughout, we use $\log(x)=\log_2(x)$ to indicate base 2 logarithm. 
For $x,y\in\R^m$, we let $x\circ y\in \R^m$ denote 
the element-wise product of the two vectors.

\section{An Accelerated Newton--Dinkelbach Method}
\label{sec:newton}

Let $f:\R\rightarrow \bar{\R}$ be a proper concave function such that $f(\delta)\leq 0$ and $\partial f(\delta)\cap \R_{<0}\neq \emptyset$ for some $\delta\in \dom(f)$.
Given a suitable starting point, as well as value and supergradient oracles of $f$, the Newton--Dinkelbach method either computes the largest root of $f$ or declares that it does not have a root.
In this paper, we make the mild assumption that $f$ has a root or attains its maximum.
Consequently, the point
\[\delta^* := \max (\set{\delta:f(\delta) = 0}\cup \argmax f(\delta))\]
is well-defined.
It is the largest root of $f$ if $f$ has a root.
Otherwise, it is the largest maximizer of $f$.
Therefore, the Newton--Dinkelbach method returns $\delta^*$ if $f$ has a root, and certifies that $f(\delta^*)<0$ otherwise.

The algorithm takes as input an initial point $\delta^{(1)}\in \dom(f)$ and a supergradient $g^{(1)}\in \partial f(\delta^{(1)})$ such that $f(\delta^{(1)})\leq 0$ and $g^{(1)}<0$.
At the start of every iteration $i\geq 1$, it maintains a point $\delta^{(i)}\in \dom(f)$ and a supergradient $g^{(i)}\in \partial f(\delta^{(i)})$ where $f(\delta^{(i)})\leq 0$.
If $f(\delta^{(i)}) = 0$, then it returns $\delta^{(i)}$ as the largest root of $f$.
Otherwise, a new point $\delta := \delta^{(i)} - f(\delta^{(i)})/g^{(i)}$ is generated.
Now, there are two scenarios in which the algorithm terminates and reports that $f$ does not have a root: (1) $f(\delta) = -\infty$; (2) $f(\delta)<0$ and $g\geq 0$ where $g\in \partial f(\delta)$ is the supergradient given by the oracle.
If both scenarios do not apply, the next point and supergradient is set to $\delta^{(i+1)} := \delta$ and $g^{(i+1)} := g$ respectively. 
Then, a new iteration begins. 

According to this update rule, observe that $g^{(i)}<0$ except possibly in the final iteration when $f(\delta^{(i)}) = 0$.
This proves the correctness of the algorithm.
Indeed, $\delta^{(i)} = \delta^*$ if $f(\delta^{(i)}) = 0$.
On the other hand, if either of the aforementioned scenarios apply, then combining it with $f(\delta^{(i)})<0$ and $g^{(i)}<0$ certifies that $f(\delta^*) < 0$.

The following lemma shows that $\delta^{(i)}$ is monotonically decreasing while $f(\delta^{(i)})$ is monotonically increasing.
Furthermore, $g^{(i)}$ is monotonically increasing except in the final iteration where it may remain unchanged.
The lemma also illustrates the useful property that $|f(\delta^{(i)})|$ or $|g^{(i)}|$ decreases geometrically.
These are well-known facts and similar statements can be found in e.g.~Radzik \cite[Lemmas 3.1 \& 3.2]{book/hco/Radzik98}. 

\begin{apxlemmarep}\label{lem:monotone}
For every iteration $i \geq 2$, we have $\delta^* \leq \delta^{(i)} < \delta^{(i-1)}$, $f(\delta^*) \geq f(\delta^{(i)}) > f(\delta^{(i-1)})$ and $g^{(i)} \geq g^{(i-1)}$, where the last inequality holds at equality if and only if $g^{(i)} = \inf_{g \in \partial f(\delta^{(i)})} g$, $g^{(i-1)} = \sup_{g \in \partial f(\delta^{(i-1)})} g$ and $f(\delta^{(i)}) = 0$. Moreover,
\[
\frac{f(\delta^{(i)})}{f(\delta^{(i-1)})} + \frac{g^{(i)}}{g^{(i-1)}} \leq 1\, .
\]
\end{apxlemmarep}

\begin{proof}
Since $f(\delta^{(i)}) \leq 0$ and $g^{(i)} < 0$, by concavity of $f$ we have
that $f(\delta) \leq f(\delta^{(i)}) + g^{(i)}(\delta-\delta^{(i)}) <
f(\delta^{(i)}) \leq 0$, for all $\delta > \delta^{(i)}$. Given this, we must have
$\delta^* \leq \delta^{(i)}$ since either $f(\delta^*) = 0 \geq
f(\delta^{(i)})$ or $0 > f(\delta^*) = \max_{z \in \R} f(z) \geq
f(\delta^{(i)})$. As $\delta^{(i)} = \delta^{(i-1)} -
\frac{f(\delta^{(i-1)})}{g^{(i-1)}} < \delta^{(i-1)}$, since
$f(\delta^{(i-1)}), g^{(i-1)} < 0$, we have $f(\delta^{(i-1)}) <
f(\delta^{(i)})$. Furthermore, $g^{(i)} \geq g^{(i-1)}$ is immediate from the
concavity of $f$.

To understand when $g^{(i)} = g^{(i-1)}$, we see by concavity that
\[
g^{(i)} \geq \inf_{g \in \partial f(\delta^{(i)})} g \geq 
\frac{f(\delta^{(i-1)})-f(\delta^{(i)})}{\delta^{(i-1)}-\delta^{(i)}}\geq
\sup_{g \in \partial f(\delta^{(i-1)})} g \geq g^{(i-1)}. 
\]
To have equality throughout, we must therefore have that $g^{(i)}$ and $g^{(i-1)}$
are equal to the respective infimum and supremum.
We must also have $f(\delta^{(i)}) = 0$ since 
\[
\frac{f(\delta^{(i-1)})-f(\delta^{(i)})}{\delta^{(i-1)}-\delta^{(i)}} 
= \frac{f(\delta^{(i-1)})-f(\delta^{(i)})}{\frac{f(\delta^{(i-1)})}{g^{(i-1)}}} = g^{(i-1)}\left(1-\frac{f(\delta^{(i)})}{f(\delta^{(i-1)})}\right)  
\]
To have equality throughout, we must therefore have that $g^{(i)}$ and $g^{(i-1)}$
are equal to the respective infimum and supremum and that $f(\delta^{(i)}) =
0$. 

Lastly, since $f$ is concave
\[
f(\delta^{(i-1)}) \leq f(\delta^{(i)}) + g^{(i)}(\delta^{(i-1)}-\delta^{(i)}) = f(\delta^{(i)}) + g^{(i)}\frac{f(\delta^{(i-1)})}{g^{(i-1)}}.\]
The moreover now follows by dividing both sides by $f(\delta^{(i-1)})< 0$.
\end{proof}

Our analysis of the Newton--Dinkelbach method utilizes the Bregman divergence associated with $f$ as a potential. 
Even though the original definition requires $f$ to be differentiable and strictly concave, it can be naturally extended to our setting in the following way.

\begin{definition}
Given a proper concave function $f:\R\rightarrow \bar{\R}$, the \emph{Bregman divergence associated with $f$} is defined as
\[D_f(\delta', \delta) := \begin{cases}
	f(\delta) + \displaystyle \sup_ {g\in \partial f(\delta)}g (\delta'-\delta) - f(\delta') &\text{ if }\delta \neq \delta',\\
	0 &\text{ otherwise.}
\end{cases}\] 
for all $\delta,\delta'\in \dom(f)$ such that $\partial f(\delta) \neq \emptyset$.
\end{definition}

Since $f$ is concave, the Bregman divergence is nonnegative.
The next lemma shows that $D_f(\delta^*,\delta^{(i)})$ is monotonically decreasing except in the final iteration where it may remain unchanged.

\begin{apxlemmarep}\label{lem:bregman_monotone}
For every iteration $i\geq 2$, we have $D_f(\delta^*,\delta^{(i)}) \leq D_f(\delta^*,\delta^{(i-1)})$ which holds at equality if and only if $g^{(i-1)} = \inf_{g \in \partial f(\delta^{(i-1)})} g$ and $f(\delta^{(i)}) = 0$.
\end{apxlemmarep}

\begin{proof}
By Lemma \ref{lem:monotone}, we know that $\delta^*\leq \delta^{(i)} <
\delta^{(i-1)}$ and $0 \geq f(\delta^{(i)}) > f(\delta^{(i-1)})$. 
Hence,
\begin{align*}
	D_f(\delta^*,\delta^{(i-1)}) &= f(\delta^{(i-1)}) + \sup_{g\in \partial f(\delta^{(i-1)})} g(\delta^* - \delta^{(i-1)}) - f(\delta^*) \\
                                   &\geq f(\delta^{(i-1)}) + g^{(i-1)}(\delta^{(i)}-\delta^{(i-1)}) + g^{(i-1)}(\delta^{*}-\delta^{(i)}) - f(\delta^*) \\ 
                                   &= 0 + g^{(i-1)}(\delta^{*}-\delta^{(i)}) - f(\delta^*) \\
                                   &\geq f(\delta^{(i)}) + g^{(i-1)}(\delta^{*}-\delta^{(i)}) - f(\delta^*) \quad \left(\text{ by concavity of $f$ }\right)\\ 
&\geq f(\delta^{(i)}) + \sup_{g \in \partial f(\delta^{(i)})} g(\delta^{*}-\delta^{(i)}) - f(\delta^*) \\
&= D_f(\delta^*,\delta^{(i)}).
\end{align*}
For the equality condition, note that the first two inequalities hold at
equality precisely when $g^{(i-1)} = \inf_{g \in \partial f(\delta^{(i-1)})}
g$ and $f(\delta^{(i)}) = 0$. If $f(\delta^{(i)})=0$, then $\delta^{(i)} =
\delta^{*}$, and hence the third inequality holds at equality as well.
\end{proof}

To accelerate this classical method, we perform an aggressive guess $\delta'=2\delta-\delta^{(i)}<\delta$ on the next point at the end of every iteration $i$.
We call this procedure \emph{look-ahead}, which is implemented on Lines \ref{line:begin_look_ahead}--\ref{line:end_look_ahead} of Algorithm \ref{alg:newton_fast}.
Let $g'\in \partial f(\delta')$ be the supergradient returned by the oracle.
If $-\infty < f(\delta')< 0$ and $g'<0$, then the next point and supergradient are set to $\delta^{(i+1)} := \delta'$ and $g^{(i+1)} := g'$ respectively as $\delta'\geq\delta^*$.
In this case, we say that look-ahead is \emph{successful} in iteration $i$.
Otherwise, we proceed as usual by taking $\delta^{(i+1)} := \delta$ and $g^{(i+1)} := g$.
It is easy to verify that Lemmas \ref{lem:monotone} and \ref{lem:bregman_monotone} also hold for Algorithm \ref{alg:newton_fast}.

\begin{algorithm}[htb!]
	\caption{\sc Look-aheadNewton}
	\label{alg:newton_fast}

	\Input{Value and supergradient oracles for a proper concave function $f$, an initial point $\delta^{(1)}\in \dom(f)$ and supergradient $g^{(1)} \in \partial f(\delta^{(1)})$ where $f(\delta^{(1)}) \leq 0$ and $g^{(1)}<0$.}
	\Output{The largest root of $f$ if it exists; report {\tt NO ROOT} otherwise.}
	\BlankLine

	$i\leftarrow 1$\;
	\While{$f(\delta^{(i)}) < 0$}{
		$\delta \leftarrow \delta^{(i)} - f(\delta^{(i)})/g^{(i)}$\;
		$g\in \partial f(\delta)$ \tcc*[r]{Empty if $f(\delta) = -\infty$}
		\If{$f(\delta) = -\infty$ \textbf{or} ($f(\delta)<0$ \textbf{and} $g\geq 0$)}{
			\Return {\tt NO ROOT}\;
		}
		$\delta' \leftarrow 2\delta - \delta^{(i)}$ \tcc*[r]{Look-ahead guess} \label{line:begin_look_ahead}
		$g'\in \partial f(\delta')$ \tcc*[r]{Empty if $f(\delta') = -\infty$}
		\If(\tcc*[f]{Is the guess successful?}){$-\infty<f(\delta') < 0$ \textbf{and} $g'<0$}{
			$\delta\leftarrow \delta'$, $g\leftarrow g'$ \label{line:end_look_ahead} \;
		}
		$\delta^{(i+1)}\leftarrow \delta$, $g^{(i+1)}\leftarrow g$\;	
		$i\leftarrow i + 1$\;		
	}
	\Return $\delta^{(i)}$\;
\end{algorithm}

If look-ahead is successful, then we have made significant progress.
Otherwise, by our choice of $\delta'$, we learn that we are not too far away from $\delta^*$.
The next lemma demonstrates the advantage of using the look-ahead Newton--Dinkelbach method. 
It exploits the proximity to $\delta^*$ to produce a geometric decay in the Bregman divergence of $\delta^{(i)}$ and $\delta^*$.

\begin{lemma}\label{lem:bregman_decay}
For every iteration $i > 2$ in Algorithm \ref{alg:newton_fast}, we have $D_f(\delta^*,\delta^{(i)}) < \frac12 D_f(\delta^*,\delta^{(i-2)})$.
\end{lemma}

\begin{proof}
Fix an iteration $i>2$ of Algorithm \ref{alg:newton_fast}. Let $g^{(i)}_+ =
\min_{g \in \partial f(\delta^{(i)})} g$ denote the right derivative of $f$
at $\delta^{(i)}$. From Lemma~\ref{lem:monotone}, we know that $\delta^* \leq
\delta^{(i)} < \delta^{(i-1)} < \delta^{(i-2)}$, $0 \geq f(\delta^*) \geq
f(\delta^{(i)}) > f(\delta^{(i-1)})
> f(\delta^{(i-2)})$ and $0 > g^{(i)}_+ \geq g^{(i-1)} > g^{(i-2)}$. Since
$\delta^* \leq \delta^{(i)}$, we see that
$D_f(\delta^*,\delta^{(i)}) = f(\delta^{(i)}) +
g^{(i)}_+(\delta^*-\delta^{(i)}) - f(\delta^*)$.    

Assume first that the look-ahead step in iteration $i-1$ was successful.  We
now claim that $0 < -2g_+^{(i)} \leq -g^{(i-1)}$. To see this, we have that
\begin{align*}
f(\delta^{(i-1)}) &\leq f(\delta^{(i)}) + g^{(i)}_+(\delta^{(i-1)}-\delta^{(i)}) \quad \left(\text{ by concavity of $f$ }\right) \\
&\leq g^{(i)}_+(\delta^{(i-1)}-\delta^{(i)}) \quad \left(\text{ since } f(\delta^{(i)}) \le 0 \right) \\
&= 2g^{(i)}_+ \frac{f(\delta^{(i-1)})}{g^{(i-1)}} \quad \left(\text{ by definition of the accelerated step }\right).
\end{align*}
The desired inequality follows by multiplying through by
$-\frac{g^{(i-1)}}{f(\delta^{(i-1)})} < 0$.

Using the above inequality, we compare Bregman divergences as follows:
\begin{align*}
D_f(\delta^{*},\delta^{(i-1)}) &\geq f(\delta^{(i-1)}) + g^{(i-1)}(\delta^*-\delta^{(i-1)}) - f(\delta^*) \quad \left(\text{ since $D_f$ is a maximum over supergradients }\right) \\
   &>  g^{(i-1)}(\delta^*-\delta^{(i)}) - f(\delta^{*}) \quad \left(~f(\delta^{(i-1)}) + g^{(i-1)}(\delta^{(i)}-\delta^{(i-1)}) = -f(\delta^{(i-1)}) > 0~\right) \\
    &\geq g^{(i-1)}(\delta^*-\delta^{(i)}) \quad \left(~-f(\delta^*) \geq 0~\right) \\
    &\geq 2g^{(i)}_+(\delta^*-\delta^{(i)}) \quad \left(~-g^{(i-1)} \geq -2g^{(i)}_+ \text{ and } \delta^{(i)} > \delta^{*}~\right) \\
    &\geq 2(f(\delta^{(i)}) + g^{(i)}_+(\delta^*-\delta^{(i)}) - f(\delta^*)) \quad \left(\text{ since } f(\delta^*) \geq f(\delta^{(i)})\right) \\
    &= 2D_f(\delta^{*}, \delta^{(i)}) \quad \left(\text{ by our choice of } g^{(i)}_+~\right).
\end{align*}
The desired inequality nows follows from $D_f(\delta^{*},\delta^{(i-2)}) > 
D_f(\delta^{*},\delta^{(i-1)})$ by Lemma~\ref{lem:bregman_monotone}.

Now assume that the look-ahead step at iteration $i-1$ was unsuccessful.
This implies that $2\delta^{(i)}-\delta^{(i-1)} \leq \delta^{*} \Leftrightarrow
2(\delta^{(i)}-\delta^*) \leq \delta^{(i-1)}-\delta^*$, i.e. that the look-ahead
step ``went past or exactly to'' $\delta^{*}$. We compare Bregman-divergences as follows:
\begin{align*}
D_f(\delta^{*},\delta^{(i-2)}) &\geq f(\delta^{(i-2)}) + g^{(i-2)}(\delta^*-\delta^{(i-2)}) - f(\delta^*) \quad \left(\text{ since $D_f$ is a maximum over supergradients }\right) \\
   &\geq  g^{(i-2)}(\delta^*-\delta^{(i-1)}) - f(\delta^{*}) \quad \left(~f(\delta^{(i-2)}) + g^{(i-2)}(\delta^{(i-1)}-\delta^{(i-2)}) \geq 0~\right) \\
    &\geq g^{(i-2)}(\delta^*-\delta^{(i-1)}) \quad \left(~-f(\delta^*) \geq 0~\right) \\
    &> g^{(i)}_+(\delta^*-\delta^{(i-1)}) \quad \left(~0 > g^{(i)}_+ > g^{(i-2)} \text{ and }  \delta^{(i-1)} > \delta^*~\right) \\
    &\geq 2g^{(i)}_+(\delta^*-\delta^{(i)}) \quad \left(~0 > g^{(i)}_+ \text{ and } \delta^{(i-1)}-\delta^* \geq 2(\delta^{(i)}-\delta^*)~\right) \\
    &\geq 2(f(\delta^{(i)}) + g^{(i)}_+(\delta^*-\delta^{(i)}) - f(\delta^*)) \quad \left(\text{ since } f(\delta^*) \geq f(\delta^{(i)})\right) \\
    &= 2D_f(\delta^{*}, \delta^{(i)}) \quad \left(\text{ by our choice of } g^{(i)}_+~\right).
\end{align*}
This concludes the proof.
\end{proof}

In the remaining of this section, we apply the accelerated Newton--Dinkelbach method to linear fractional combinatorial optimization and linear fractional programming.
The application to parametric submodular function minimization is in Section \ref{sec:sfm}.

\subsection{Linear Fractional Combinatorial Optimization}
The problem \eqref{prob:minratio-D} with $\mathcal{D}\subseteq \{0,1\}^m$ is known as
\emph{linear fractional combinatorial optimization}.  
Radzik \cite{conf/focs/Radzik92} showed that the Newton--Dinkelbach method applied to the function $f(\delta)$ as in \eqref{prob:f-D} terminates in a strongly polynomial number of iterations. Recall that $f(\delta) = \min_{x\in\mathcal{D}} (c-\delta d)^{\top}x$.
By the assumption $d^\top x>0$ for all $x\in \mathcal{D}$, this function is 
 concave, strictly decreasing, finite and piecewise-linear.
Hence, it has a unique root.
Moreover, $f(\delta)<0$ and $\partial f(\delta)\cap \R_{<0} \neq \emptyset$ for sufficiently large $\delta$.
To implement the value and supergradient oracles, we assume that a linear optimization oracle over $\mathcal{D}$ is available, i.e.~it returns an element in $\argmin_{x\in \mathcal{D}} (c-\delta d)^{\top}x$ for any $\delta\in \R$.

Our result for the accelerated variant improves the state-of-the-art bound $O(m^2\log m)$ by Wang et al.~\cite{journals/jgo/WangYZ06} on the standard Newton--Dinkelbach method.
We will need the following lemma, given by Radzik and credited to Goemans in \cite{book/hco/Radzik98}.
It gives a strongly polynomial bound on the length of a geometrically decreasing sequence of sums.

\begin{lemma}[\cite{book/hco/Radzik98}]\label{lem:goemans_m}
Let $c\in \R^m_+$ and $x^{(1)},x^{(2)},\dots,x^{(k)}\in \set{-1,0,1}^m$. If $0 < c^{\top}x^{(i+1)} \leq \frac{1}{2}c^{\top}x^{(i)}$ for all $i<k$, then $k = O(m\log m)$.
\end{lemma}

\begin{theorem}
Algorithm \ref{alg:newton_fast} converges in $O(m\log m)$ iterations for linear fractional combinatorial optimization problems. 
\end{theorem}

\begin{proof}
Observe that Algorithm \ref{alg:newton_fast} terminates in a finite number of iterations because $f$ is piecewise linear.
Let $\delta^{(1)}> \delta^{(2)}> \dots>\delta^{(k)} = \delta^{*}$ denote the sequence of 
iterates at the start of Algorithm \ref{alg:newton_fast}.
Since $f$ is concave, we have $D_f(\delta^*,\delta^{(i)}) \geq 0$ for all $i \in [k]$. 
For each $i \in [k]$, pick $x^{(i)} \in \arg\min_{x\in \mathcal{D}} (c-\delta^{(i)}d)^{\top}x$ which maximizes $d^{\top}x$.
This is well-defined because $f$ is finite.
Note that $-d^{\top}x^{(i)} = \min \partial f(\delta^{(i)})$.
As $f(\delta^*)=0$, the Bregman divergence of $\delta^{(i)}$ and $\delta^*$ can be written as
\[
D_f(\delta^*,\delta^{(i)})=f(\delta^{(i)})+\max_{g\in \partial f(\delta^{(i)})}g(\delta^*-\delta^{(i)})=(c-\delta^{(i)}d)^\top x^{(i)}-d^\top x^{(i)} (\delta^*-\delta^{(i)})=(c-\delta^{*}d)^\top x^{(i)}\, .
\]
According to Lemma \ref{lem:bregman_decay},
$(c-\delta^* d)^{\top}x^{(i)} = D_f(\delta^*,\delta^{(i)}) < \frac12 D_f(\delta^*,\delta^{(i-2)}) = \frac12 (c-\delta^* d)^{\top}x^{(i-2)}$ for all $3 \leq i \leq k$. 
By Lemma \ref{lem:bregman_monotone}, we also know that $D_f(\delta^*,\delta^{(i)})> 0$ for all $1\leq i\leq k-2$.
Thus, applying Lemma \ref{lem:goemans_m} yields $k = O(m\log m)$.
\end{proof}

\subsection{Linear Fractional Programming}
We next consider \emph{linear fractional programming}, an extension of \eqref{prob:minratio-D} with the assumption that the domain $\mathcal{D}\subseteq \R^m$ is a polyhedron, but removing the condition $d^\top x>0$ for $x\in\mathcal{D}$. For $c,d\in \R^m$, the problem is
\begin{equation}
	\label{sys_F}
	\tag{F}
	\inf c^{\top}x/d^{\top}x\quad \mbox{s.t. }d^{\top}x >0,\; x\in \mathcal{D}\, .
\end{equation}
For the problem to be meaningful, we assume that $\mathcal{D}\cap \set{x:d^{\top}x>0} \neq \emptyset$.
The common form in the literature assumes $d^{\top}x >0$ for all $x\in \mathcal{D}$ as in \eqref{prob:minratio-D};  
we consider the more general setup for the purpose of solving M2VPI systems in Section \ref{sec:2vpi}.
It is easy to see that any linear fractional combinatorial optimization problem on a domain $\mathcal{X}\subseteq \{0,1\}^m$ can be cast as a linear fractional program with the polytope $\mathcal{D} = \conv(\mathcal{X})$ because $c^{\top}\bar{x}/d^{\top}\bar{x}\geq \min_{x\in \mathcal{X}}c^{\top}x/d^{\top}x$ for all $\bar{x}\in \mathcal{D}$.
The next theorem characterizes when \eqref{sys_F} is  unbounded.

\begin{theorem}\label{thm:linfrac}
If $\mathcal{D}\cap \set{x:d^{\top}x>0} \neq \emptyset$, then the optimal value of $\eqref{sys_F}$ is $-\infty$ if and only if at least one of the following two conditions hold:
\begin{enumerate}%
	\item There exists $x\in \mathcal{D}$ such that $c^{\top}x<0$ and $d^{\top}x = 0$;
	\item There exists $r\in \R^m$ such that $c^{\top}r <0$, $d^{\top}r = 0$ and $x+\lambda r\in \mathcal{D}$ for all $x\in \mathcal{D}, \lambda\geq 0$.
\end{enumerate}
\end{theorem}

\begin{proof}
	By the Minkowski-Weyl Theorem, the polyhedron $\bar{\mathcal{D}} := \mathcal{D} \cap \{x : d^\top x \ge 0\}$ can be written as 
	\begin{equation*}
		\bar {\mathcal{D}} = \left\{\sum_{i = 1}^{k} \lambda_i g_i + \sum_{j = 1}^{\ell} \nu_j h_j : \lambda \ge 0, \nu \ge 0, \|\lambda\|_1 = 1\right\}
	\end{equation*}
	for some vectors $g_1, \ldots, g_k$ and $h_1, \ldots, h_\ell$.
	Note that $d^\top g_i\ge 0$ for all $i \in [k]$ and $d^\top h_j\ge 0$ for all $j \in [\ell]$. 
	Let $x^\circ \in \mathcal{D} \cap \{x : d^\top x > 0$\}. 
	If there exists $i \in [k]$ such that $c^\top g_i < 0$ and $d^\top g_i = 0$ or $j \in [\ell]$ such that $c^\top h_j < 0$ and $d^\top h_j = 0$, then, 
	\begin{equation*}
	 \lim_{\lambda \nearrow 1} \frac{c^\top(\lambda g_i + (1-\lambda)x^\circ)}{d^\top(\lambda g_i + (1-\lambda)x^\circ)} = -\infty \quad \text{or}  \quad 
	 \lim_{\lambda \rightarrow \infty} \frac{c^\top(x^\circ + \lambda h_j)}{d^\top(x^\circ + \lambda h_j)} = -\infty
	\end{equation*}
	as in Condition 1 or Condition 2.

	Otherwise, the fractional value of any element in $\mathcal{D} \cap \{x : d^\top x > 0\}$ can be lower bounded by 
	\begin{equation*}
		\begin{aligned}
		\frac{c^\top(\sum_{i = 1}^{k} \lambda_i g_i + \sum_{j = 1}^{\ell} \nu_j h_j)}{d^\top(\sum_{i = 1}^{k} \lambda_i g_i + \sum_{j = 1}^{\ell} \nu_j h_j)}
		& \ge \frac{\sum_{i \in [k], d^\top g_i > 0} \lambda_i c^\top g_i + \sum_{j \in [\ell], d^\top h_j > 0} \nu_j c^\top h_j}
		{\sum_{i \in [k], d^\top g_i > 0} \lambda_i d^\top g_i + \sum_{j \in [\ell], d^\top h_j > 0} \nu_j d^\top h_j} \\
		& \ge \min\left\{\min_{i \in [k], d^\top g_i > 0} \frac{c^\top g_i}{d^\top g_i}, \min_{j \in [\ell], d^\top h_j > 0} \frac{c^\top h_j}{d^\top h_j}\right\},
		\end{aligned}
	\end{equation*}
	where the last expression is finite by the assumption that $\mathcal{D} \cap \{x : d^\top x > 0\}$ is non-empty.
\end{proof}

\begin{example}
Unlike in linear programming, the optimal value may not be attained even if it is finite.
Consider the instance given by $\inf (-x_1 + x_2)/(x_1 + x_2)$ subject to $x_1 + x_2 >0$ and $-x_1 + x_2 = 1$.
The numerator is equal to 1 for any feasible solution, while the denominator can be made arbitrarily large.
Hence, the optimal value of this program is 0, which is not attained in the feasible region.
\end{example}

We use the Newton--Dinkelbach method for $f$ as in \eqref{prob:f-D}, that is, $f(\delta)= \inf_{x\in \mathcal{D}} (c-\delta d)^{\top}x  $.
Since $\mathcal{D}\neq \emptyset$, $f(\delta)<\infty$ for all $\delta\in \R$.
By the Minkowski--Weyl theorem, there exist finitely many points $P\subseteq \mathcal{D}$ such that $f(\delta) = \min_{x\in P}(c-\delta d)^{\top}x$ for all $\delta\in \dom(f)$.
Hence, $f$ is concave and piecewise linear.
Observe that $f(\delta)>-\infty$ if and only if every ray $r$ in the recession cone of $\mathcal{D}$ satisfies $(c-\delta d)^{\top}r\geq 0$.
For $f$ to be proper, we need to assume that Condition 2 in Theorem \ref{thm:linfrac} does not hold.
Moreover, we require the existence of a point $\delta'\in \dom(f)$ such that $f(\delta') = (c - \delta' d)^{\top}x'\leq0$ for some $x'\in \mathcal{D}$ with $d^{\top}x'>0$.
It follows that $f$ has a root or attains its maximum because $\dom(f)$ is closed. 
We are ready to characterize the optimal value of \eqref{sys_F} using $f$.

\begin{lemma}\label{lem:linfrac}
Assume that there exists $\delta'\in \dom(f)$ such that $f(\delta') = (c-\delta'd)^{\top}x'\leq 0$ for some $x'\in \mathcal{D}$ with $d^{\top}x'>0$. 
If $f$ has a root, then the optimal value of \eqref{sys_F} is equal to the largest root and is attained.
Otherwise, the optimal value is $-\infty$. 
\end{lemma}

\begin{proof}
Recall the definition of $\delta^* = \max(\set{\delta:f(\delta)=0}\cup\argmax f(\delta))$.
By our assumption on $f$, there exists $x^*\in \mathcal{D}$ such that $f(\delta^*) = (c-\delta^*d)^{\top}x^*$ and $d^{\top}x^*>0$.
If $f$ has a root, then $f(\delta^*) = 0$.
This implies that $c^{\top}x/d^{\top}x\geq \delta^* = c^{\top}x^*/d^{\top}x^*$ for all $x\in \mathcal{D}$ with $d^{\top}x>0$ as desired.
Next, assume that $f$ does not have a root. Then $f(\delta^*) < 0$ and $0 \in \partial f(\delta^*)$. By convexity, there exists $\bar x \in \mathcal D$ such that $(c - \delta ^* d)^\top \bar x = f(\delta^*) < 0$ and $d^\top \bar x = 0$. Then $c^\top \bar x < 0$, so $\bar x$ is a point as in Condition 1 of \Cref{thm:linfrac}.

\end{proof}

\section{Monotone Two Variable per Inequality Systems}
\label{sec:2vpi}
Recall that an M2VPI system can be represented as a directed multigraph $G=(V,E)$ with arc costs $c\in \R^m$ and gain factors $\gamma\in \R^m_{++}$.
For a $u$-$v$ walk $P$ in $G$ with $E(P)=(e_1,e_2,\dots,e_k)$, its \emph{cost} and \emph{gain factor} are defined as
$c(P) := \sum_{i=1}^k\pr{\prod_{j=1}^{i-1}\gamma_{e_j}}c_{e_i}$ and $\gamma(P) := \prod_{i=1}^k\gamma_{e_i}$ respectively. 
If $P$ is a single vertex, 
then $c(P) := 0$ and $\gamma(P) := 1$.  
The walk $P$ induces the valid inequality $y_u\leq c(P) + \gamma(P)y_v$, implied by the sequence of arcs/inequalities in $E(P)$.
It is also worth considering the dual interpretation. Dual variables on arcs correspond to generalized flows: if 1 unit of flow enter the arc $e=(u,v)$ at $u$, then $\gamma_e$ units reach $v$, at a shipping cost of $c_e$. Thus, if 1 unit of flow enter a path $P$, then $\gamma(P)$ units reach the end of the path, incurring a cost of $c(P)$.

Given node labels $y\in \bar{\R}^n$, the \emph{$y$-cost} of a $u$-$v$ walk $P$ is defined as $c(P) + \gamma(P)y_v$.
Note that the $y$-cost of a walk only depends on the label at the sink.
A $u$-$v$ path is called a \emph{shortest $u$-$v$ path with respect to $y$} if it has the smallest $y$-cost among all $u$-$v$ walks. 
A \emph{shortest path from $u$ with respect to $y$} is a shortest $u$-$v$ path with respect to $y$ for some node $v$.
Such a path does not always exist, as demonstrated in Appendix
\ref{subsec:existence_shortest_cycles_or_paths}.

If $P$ is a $u$-$u$ walk such that its intermediate nodes are distinct, then it is called a \emph{cycle at $u$}. 
Given a $u$-$v$ walk $P$ and a $v$-$w$ walk $Q$, we denote $PQ$ as the $u$-$w$ walk obtained by concatenating $P$ and $Q$.

\begin{definition}
A cycle $C$ is called \emph{flow-generating} if $\gamma(C)>1$, \emph{unit-gain} if $\gamma(C) = 1$, and \emph{flow-absorbing} if $\gamma(C)<1$. We say that a unit-gain cycle $C$ is \emph{negative} if $c(C)<0$.
\end{definition}

Note that $c(C)$ depends on the starting point $u$ of a cycle $C$. 
This ambiguity is resolved by using the term \emph{cycle at $u$}.
For a unit-gain cycle $C$, it is not hard to see that the starting point does not affect the sign of $c(C)$.
Hence, the definition of a negative unit-gain cycle is sound.
Observe that a flow-absorbing cycle $C$ induces an upper bound $y_u\leq c(C)/(1-\gamma(C))$, while a flow-generating cycle $C$ induces a lower bound $y_u\geq -c(C)(\gamma(C)-1)$.
Let $\mathcal{C}^{abs}_u(G)$ and $\mathcal{C}^{gen}_u(G)$ denote the set of flow-absorbing cycles and flow-generating cycles at $u$ in $G$ respectively.

\begin{definition}
Given a flow-generating cycle $C$ at $u$, a flow-absorbing cycle $D$ at $v$, and a $u$-$v$ path $P$, the walk $CPD$ is called a \emph{bicycle}. 
We say that the bicycle is \emph{negative} if
\[c(P) + \gamma(P)\frac{c(D)}{1-\gamma(D)} < \frac{-c(C)}{\gamma(C)-1}\,.\]
\end{definition}

Using these two structures, Shostak characterized the feasibility of M2VPI systems.

\begin{theorem}[\cite{journals/jacm/Shostak81}]\label{thm:feas_char}
An M2VPI system $(G,c,\gamma)$ is infeasible if and only if $G$ contains a negative unit-gain cycle or a negative bicycle.
\end{theorem}

\subsection{A linear fractional programming formulation}

Our goal is to compute the pointwise maximal solution $y^{\max}\in \bar{\R}^n$ to an M2VPI system if it is feasible, where $y^{\max}_u := \infty$ if and only if the variable $y_u$ is unbounded from above.
It is well known how to convert $y^{\max}$ into a finite feasible solution --- we refer to Appendix \ref{sec:2vpi_finite} for details.
In order to apply Algorithm \ref{alg:newton_fast}, we first need to reformulate the problem as a linear fractional program.
Now, every coordinate $y^{\max}_u$ can be expressed as the following primal-dual pair of linear programs, where $\nabla x_v := \sum_{e\in \delta^+(u)}x_e - \sum_{e\in \delta^-(u)}\gamma_e x_e$ denotes the net flow at a node $v$.

\begin{minipage}[t]{.48\textwidth} \vspace{-3mm}
\begin{align*}\tag{P$_u$}\label{sys:2vpi_primal}
   &\min\; c^{\top}x \\
   &\subto\;\, \nabla x_u = 1 \\
   &\qquad\; \nabla x_v = 0 \qquad \forall v\in V\setminus u \\
   &\qquad\quad\;\; x \geq 0 
\end{align*}
\end{minipage}
\begin{minipage}[t]{.48\textwidth} \vspace{-3mm}
\begin{align*}\tag{D$_u$}\label{sys:2vpi_dual}
	 &\max\; y_u \\
	 &\subto\;\; y_v - \gamma_e y_w \leq c_e \qquad \forall e=(v,w)\in E \\
\end{align*}
\end{minipage} \vspace{3mm} 

The primal LP \eqref{sys:2vpi_primal} is a minimum-cost generalized flow problem with a supply of 1 at node $u$.
It asks for the cheapest way to destroy one unit of flow at $u$.
Observe that it is feasible if and only if $u$ can reach a flow-absorbing cycle in $G$.
If it is feasible, then it is unbounded if and only if there exists a negative unit-gain cycle or a negative bicycle in $G$. 
It can be reformulated as the following linear fractional program
\begin{equation}\tag{F$_u$}\label{sys:2vpi_frac}
   \inf\; \frac{c^{\top}x}{1-\sum_{e\in \delta^-(u)}\gamma_ex_e}\quad \mbox{s.t. } 1-\sum_{e\in \delta^-(u)}\gamma_e x_e > 0,\; x\in \mathcal{D}\, .
\end{equation}
with the polyhedron 
\[\mathcal{D} := \set{x\in \R^m_+: x(\delta^+(u)) = 1, \nabla x_v = 0 \;\forall v\in V\setminus u}.\]
Indeed, if $x$ is a feasible solution to $\eqref{sys:2vpi_primal}$, then $x/x(\delta^+(u))$ is a feasible solution to \eqref{sys:2vpi_frac} with the same objective value.
This is because $1-\sum_{e\in \delta^-(u)}\gamma_e x_e/x(\delta^+(u)) = 1/x(\delta^+(u))$.
Conversely, if $x$ is a feasible solution to \eqref{sys:2vpi_frac}, then $x/(1-\sum_{e\in \delta^-(u)}\gamma_e x_e)$ is a feasible solution to \eqref{sys:2vpi_primal} with the same objective value.
Even though the denominator is an affine function of $x$, it can be made linear to conform with \eqref{sys_F} by working with the polyhedron $\set{(x,1):x\in \mathcal{D}}$. 

Our goal is to solve \eqref{sys:2vpi_frac} using Algorithm \ref{alg:newton_fast}.
Due to the specific structure of this linear fractional program, a suitable initial point for the Newton--Dinkelbach method can be obtained from any feasible solution to \eqref{sys:2vpi_frac}.
This is a consequence of the unboundedness test given by the following lemma.

\begin{lemma}\label{lem:2vpi_init}
Let $x$ be a feasible solution to \eqref{sys:2vpi_frac} and $\bar{\delta} := c^{\top}x/(1-\sum_{e\in \delta^-(u)}\gamma_e x_e)$. 
If either $f(\bar{\delta}) = -\infty$ or $f(\bar{\delta}) = c^{\top}\bar{x} - \bar{\delta}(1-\sum_{e\in \delta^-(u)}\gamma_e\bar{x}_e)<0$ for some $\bar{x}\in \mathcal{D}$ with $1-\sum_{e\in \delta^-(u)}\gamma_e \bar{x}_e \leq 0$, then the optimal value of \eqref{sys:2vpi_frac} is $-\infty$. 
\end{lemma}

\begin{proof}
First, assume that $f(\bar{\delta})>-\infty$.
Let $\lambda:= (1-\sum_{e\in \delta^-(u)}\gamma_e x_e) / \sum_{e\in \delta^-(u)}\gamma_e(\bar{x}_e-x_e)$.
Note that $\lambda\in (0,1]$.
Consider the convex combination $\hat{x}:= \lambda\bar{x} + (1-\lambda)x\in \mathcal{D}$.
Then, $c^{\top}\hat{x}<0$ and $1-\sum_{e\in \delta^-(u)}\gamma_e\hat{x}_e = 0$.
Hence, the optimal value of \eqref{sys:2vpi_frac} is unbounded by Condition 1 of Theorem \ref{thm:linfrac}.
Next, assume that $f(\bar{\delta}) = -\infty$.
There exists a ray $r$ in the recession cone of $\mathcal{D}$ such that $c^{\top}r - \bar{\delta}\sum_{e\in \delta^-(u)}\gamma_er_e<0$.
Note that $r\geq 0$.
If $r(\delta^-(u)) = 0$, then $r$ satisfies Condition 2 of Theorem \ref{thm:linfrac}. 
So, the optimal value is unbounded.
Otherwise, for a sufficiently large $\alpha>0$, we have $c^{\top}(x+\alpha r) + \bar{\delta}(1-\sum_{e\in \delta^-(u)}\gamma_e(x_e+\alpha r_e)) < 0$ and $1-\sum_{e\in \delta^-(u)}\gamma_e(x_e+\alpha r_e)<0$.
Then, taking an appropriate convex combination of $x+\alpha r$ and $x$ like before produces a point in $\mathcal{D}$ which satisfies Condition 1 of Theorem \ref{thm:linfrac}.
\end{proof}

For a fixed $\delta\in \R$, the value of the parametric function $f(\delta)$ can be written as the following pair of primal and dual LPs respectively

\begin{minipage}[t]{.43\textwidth} \vspace{-3mm}
\begin{align*}
   &\min\; c^{\top}x+\delta\sum_{e\in \delta^-(u)}\gamma_e x_e - \delta \\
   &\subto\;\, x\in \mathcal{D}
\end{align*}
\end{minipage}
\begin{minipage}[t]{.53\textwidth} \vspace{-3mm}
\begin{align*}
	 &\max\; y_u - \delta \\
	 &\subto\quad\, y_v - \gamma_e\delta \leq c_e \qquad \forall e = (v,u)\in \delta^-(u) \\
	 &\qquad\, y_v - \gamma_e y_w \leq c_e \qquad \forall e=(v,w)\notin \delta^-(u).
\end{align*}
\end{minipage} \vspace{3mm} 

We refer to them as the \emph{primal (resp.~dual) LP for $f(\delta)$}, and their corresponding feasible solution as a \emph{feasible primal (resp.~dual) solution to $f(\delta)$}.
In order to characterize the finiteness of $f(\delta)$, we introduce the following notion of a negative flow-generating cycle.

\begin{definition}
For a fixed $\delta\in \R$ and $u\in V$, a flow-generating cycle $C$ is said to be \emph{$(\delta,u)$-negative} if there exists a path $P$ from a node $v\in V(C)$ to node $u$ such that 
\[c(C) + (\gamma(C)-1)(c(P) + \gamma(P)\delta) < 0\] %
where $C$ is treated as a $v$-$v$ walk in $c(C)$. 
\end{definition}

\begin{lemma}\label{lem:2vpi_infty}
For any $\delta\in \R$, $f(\delta) = -\infty$ if and only if $\mathcal{D} \neq \emptyset$ and there exists a negative unit-gain cycle, a negative bicycle, or a $(\delta,u)$-negative flow-generating cycle in $G\setminus \delta^+(u)$.
\end{lemma}

\begin{proof}
The primal LP for $f(\delta)$ is unbounded if and only if $\mathcal{D}\neq \emptyset$ and there exists an extreme ray $r$ in the recession cone of $\mathcal{D}$ such that $c^{\top}r + y_u\sum_{e\in \delta^-(u)}\gamma_e r_e < 0$.
Note that the recession cone of $\mathcal{D}$ is $\set{x\in \R^m_+: x(\delta^+(u))=0, \nabla x_v = 0)\; \forall v\neq u}$.
By the generalized flow decomposition theorem, $r$ belongs to one of the following three fundamental flows in $G\setminus \delta^+(u)$: (1) a unit-gain cycle, (2) a bicycle, (3) a flow-generating cycle $C$ and a path $P$ from $C$ to $u$.
In the first two cases, $r_e = 0$ for all $e\in \delta^-(u)$.
Thus, the unit-gain cycle or bicycle is negative.
In the last case, we have $c(C) + (\gamma(C)-1)(c(P) + \gamma(P)\delta) = c^{\top}r + \delta\sum_{e\in \delta^-(u)}\gamma_e r_e$.
\end{proof}

It turns out that if we have an optimal dual solution $y$ to $f(\delta)$ for some $\delta\in \R$, then we can compute an optimal dual solution to $f(\delta')$ for any $\delta' < \delta$.
A suitable subroutine for this task is the so called {\sc Grapevine} algorithm (Algorithm~\ref{alg:grapevine}), developed by Aspvall and Shiloach \cite{journals/siamcomp/AspvallS80}.

\begin{algorithm}[H]
	\caption{\textsc{Grapevine}}
	\label{alg:grapevine}

	\Input{A directed multigraph $G=(V,E)$ with arc costs $c\in \R^m$ and gain factors $\gamma\in \R^m_{++}$, node labels $y\in \bar{\R}^n$, and a node $u\in V$.}
	\Output{Node labels $y\in \bar{\R}^n$ and a walk $P$ of length at most $n$ starting from $u$.}
	\BlankLine
	\For{$i=1$ \KwTo $n$}{
		\ForEach{$v\in V$}{
			$y'_v \leftarrow \min(y_v, \min_{vw\in \delta^+(v)}c_{vw} + \gamma_{vw} y_w)$\;
			\uIf{$y'_v < y_v$}{
				$\pred(v, i) \leftarrow \argmin_{vw\in \delta^+(v)}c_{vw} + \gamma_{vw} y_w$ \tcc*[r]{Break ties}
			}
			\Else{
				$\pred(v,i) \leftarrow \emptyset$\;
			}
		}
		$y\gets y'$\;
	}
	Let $P$ be the walk obtained by tracing from $\pred(u,n)$ \label{line:trace}\;
	\Return $(y,P)$\;
\end{algorithm}

Given initial node labels $y\in \bar{\R}^n$ and a specified node $u$, {\sc Grapevine} runs for $n$ iterations. 
We say that an arc $e=(v,w)$ is \emph{violated with respect to $y$} if $y_v> c_e + \gamma_e y_w$.
In an iteration $i\in [n]$, the algorithm records the most violated arc with respect to $y$ in $\delta^+(v)$ as $\pred(v,i)$, for each node $v\in V$ (ties are broken arbitrarily).
Note that $\pred(v,i) = \emptyset$ if there are no violated arcs in $\delta^+(v)$.
Then, each $y_v$ is decreased by the amount of violation in the corresponding recorded arc. 
After $n$ iterations, the algorithm traces a walk $P$ from $u$ by following the recorded arcs in reverse chronological order.
During the trace, if $\pred(v,i) = \emptyset$ for some $v\in V$ and $i>1$, then $\pred(v,i-1)$ is read.
Finally, the updated node labels $y$ and the walk $P$ are returned.
Clearly, the running time of {\sc Grapevine} is $O(mn)$.

Given an optimal dual solution $y\in \R^n$ to $f(\delta)$ and $\delta'< \delta$, the dual LP for $f(\delta')$ can be solved using {\sc Grapevine} as follows.
Define the directed graph $G_u := (V\cup\set{u'}, E_u)$ where $E_u := (E\setminus \delta^-(u)) \cup \set{vu':vu\in \delta^-(u)}$.
The graph $G_u$ is obtained from $G$ by splitting $u$ into two nodes $u,u'$ and reassigning the incoming arcs of $u$ to $u'$.
These arcs inherit the same costs and gain factors from their counterparts in $G$.
Let $\bar{y}\in \R^{n+1}$ be  node labels in $G_u$ defined by $\bar{y}_{u'} := \delta'$ and $\bar{y}_v := y_v$ for all $v\neq u'$.
Then, we run {\sc Grapevine} on $G_u$ with input node labels $\bar{y}$ and node $u$.
Note that $\bar{y}_{u'}$ remains unchanged throughout the algorithm.
The next lemma verifies the correctness of this method.

\begin{lemma}\label{lem:grapevine}
Given an optimal dual solution $y\in \R^n$ to $f(\delta)$ and $\delta'< \delta$, define $\bar{y}\in \R^{n+1}$ as $\bar{y}_{u'} := \delta'$  and $\bar{y}_v := y_v$ for all $v\in V$.
Let $(\bar{z},P)$ be the node labels and walk returned by {\sc Grapevine}$(G_u,\bar{y},u)$.
If $\bar{z}_V$ is not feasible to the dual LP for $f(\delta')$, then $f(\delta') = -\infty$.
Otherwise, $\bar{z}_V$ is a dual optimal solution to $f(\delta')$ and $P$ is a shortest path from $u$ with respect to $\bar{y}$ in $G_u$.
\end{lemma}

\begin{proof}
Since $f(\delta) = y_u - \delta$ is finite, we have $\mathcal{D}\neq \emptyset$.
First, assume that $\bar{z}_V$ is not feasible to the dual LP for $f(\delta')$. 
Then, there exists a violated arc in $G_u$ with respect to $\bar{z}$.
Let $w$ be the head of this arc and let $R$ be the walk obtained by tracing $\pred(w,n)$ in reverse chronological order. 
Then, $R$ ends at $u'$ because $y$ is dual feasible to $f(\delta)$. 
Since $R$ has $n$ edges, decompose it into $R=QCP'$ where $Q$ is a $w$-$v$ walk, $C$ is a nontrivial cycle at $v$, and $P$ is a $v$-$u'$ path for some node $v$. 
Then, we have $c(CP') + \gamma(CP')\delta'< c(P') + \gamma(P')\delta' \leq \bar{y}_v$. 
Due to Lemma \ref{lem:2vpi_infty}, it suffices to show that $\gamma(C)>1$, as this would imply that $C$ is a $(\delta',u)$-negative flow-generating cycle in $G$.
Suppose otherwise for a contradiction.
Since $y$ is dual feasible to $f(\delta)$ and $u'\notin V(C)$, we have $\bar{y}_v \leq c(C) + \gamma(C)\bar{y}_v$.
If $\gamma(C)=1$, then we obtain $0 \leq c(C) < 0$ from the previous two inequalities. 
Otherwise, we get the following contradiction
\[\bar{y}_v \leq \frac{c(C)}{1-\gamma(C)} < c(P') + \gamma(P')\delta' \leq \bar{y}_v.\]

Next, assume that $\bar{z}_V$ is a dual feasible solution to $f(\delta')$.
Then, $P$ is a $u$-$t$ path for some node $t$. 
This is because if $P$ is not simple, repeating the argument from the previous paragraph proves that the dual LP for $f(\delta')$ is infeasible.
Note that $\bar{y}_t = \bar{z}_t$.
Moreover, $\bar{z}_v \leq c_{vw} + \gamma_{vw}\bar{z}_w$ for all $vw\in E_u$, with equality on $E(P)$.
Let $c^{\bar{z}}\in \R^m_+$ be the reduced cost defined by $c^{\bar{z}}_{vw} := c_{vw} + \gamma_{vw}\bar{z}_w -\bar{z}_v$ for all $vw\in E_u$.
Since for every $u$-$t$ walk $P'$ we have
\[c(P) + \gamma(P)\bar{z}_t - \bar{z}_u = c^{\bar{z}}(P) = 0\leq c^{\bar{z}}(P') = c(P') + \gamma(P')\bar{z}_t - \bar{z}_u,\]
it follows that $P$ is a shortest $u$-$t$ path with respect to $\bar{y}$.

It is left to show that $\bar{z}_V$ is a dual optimal solution to $f(\delta')$.
Let $z^*$ be an optimal dual solution to $f(\delta')$.
Note that $z^*_u\leq y_u$ because $\delta'< \delta$.
For the purpose of contradiction, suppose that $\bar{z}_u<z^*_u$.
Since $\bar{z}_u < \bar{y}_u$, the path $P$ ends at $u'$ because $y$ is dual feasible to $f(\delta)$.
Thus, $\bar{z}_u = c(P) + \gamma(P)\delta'$.
However, $P$ also implies the valid inequality $z^*_u\leq c(P) + \gamma(P)\delta'$, which is a contradiction.
\end{proof}

If $\bar{z}_V$ is an optimal dual solution to $f(\delta')$, a supergradient in $\partial f(\delta')$ can be inferred from the returned path $P$.
We say that an arc $e=(v,w)$ is \emph{tight with respect to $\bar{z}$} if $\bar{z}_v= c_e + \gamma_e\bar{z}_w$.
By complementary slackness, every optimal primal solution to $f(\delta')$ is supported on the subgraph of $G_u$ induced by tight arcs with respect to $\bar{z}$.
In particular, any $u$-$u'$ path or any path from $u$ to a flow-absorbing cycle in this subgraph constitutes a basic optimal primal solution to $f(\delta')$.
As $P$ is also a path in this subgraph, we have $\gamma(P)-1\in \partial f(\delta')$ if $P$ ends at $u'$.
Otherwise, $u$ can reach a flow-absorbing cycle in this subgraph because $\delta'<\delta$.
In this case, $-1\in \partial f(\delta')$.

\subsection{A Strongly Polynomial Label-Correcting Algorithm}
\label{sec:main}

Using Algorithm \ref{alg:newton_fast}, we develop a strongly polynomial label-correcting algorithm for solving an M2VPI system $(G,c,\gamma)$.
The main idea is to start with a subsystem for which \eqref{sys:2vpi_dual} is trivial, and progressively solve \eqref{sys:2vpi_dual} for larger and larger subsystems.
Throughout the algorithm, we maintain node labels $y\in \bar{\R}^n$ which form valid upper bounds on each variable. 
They are initialized to $\infty$ at every node. 
We also maintain a subgraph of $G$, which initially is $G^{(0)} := (V,\emptyset)$. 

\begin{algorithm}[htb!]
	\caption{Label-correcting algorithm for M2VPI systems}
	\label{alg:main}

	\Input{An M2VPI system $(G,c,\gamma)$.}
	\Output{The pointwise maximal solution $y^{\max}$ or the string {\tt INFEASIBLE}.}
	\BlankLine

	Initialize graph $G^{(0)} \gets (V,\emptyset)$ and counter $k\gets 0$\;
	Initialize node labels $y\in \bar{\R}^n$ as $y_v \gets \infty \;\forall v\in V$\;

  	\ForEach{$u\in V$}{
  		$k\gets k + 1$\;
		$G^{(k)} \leftarrow G^{(k-1)} \cup \delta^+(u)$\;
		$y_u\leftarrow \min_{uv\in \delta^+(u)}c_{uv} + \gamma_{uv}y_v$\;
		\If{$y_u = \infty$ \textbf{and} $\mathcal{C}^{abs}_u(G^{(k)}) \neq \emptyset$}{
			$y_u \gets c(C)/(1-\gamma(C))$ for any $C\in \mathcal{C}^{abs}_u(G^{(k)})$
		}
		\If{$y_u < \infty$ \label{line:label_correcting}}{
			Define node labels $\bar{y}\in \bar{\R}^{n+1}$ as $\bar{y}_{u'} \gets y_u$ and $\bar{y}_v \gets y_v \;\forall v\in V$\;
			$(\bar{y},P)\leftarrow$ \textsc{Grapevine($G^{(k)}_u$,$\bar{y}$,$u$)}\;
			\If{$\exists$ a violated arc w.r.t.~$\bar{y}$ in $G^{(k)}_u$  \textbf{or} ($\size{E(P)}>0$ \textbf{and} $\gamma(P) \geq 1 $)}{
				\Return {\tt INFEASIBLE}
			}
			$\bar{y}_{u'}\gets${\sc Look-aheadNewton}({\sc Grapevine}$(G^{(k)}_u,\cdot,u), \bar{y}_{u'}, \gamma(P)-1)$\; \label{line:newton}
			\If{$\bar{y}_{u'} = $ {\tt NO ROOT}}{
				\Return {\tt INFEASIBLE}
			
			}
		$y\gets \bar{y}_V$\;
		}

  	}
	\Return $y$\;

\end{algorithm}

The algorithm (Algorithm~\ref{alg:main}) is divided into $n$ \emph{phases}. 
At the start of phase $k\in [n]$, a new node $u\in V$ is selected and all of its outgoing arcs in $G$ are added to $G^{(k-1)}$, resulting in a larger subgraph $G^{(k)}$. 
Since $y_u = \infty$ at this point, we update it to the smallest upper bound implied by its outgoing arcs and the labels of its outneighbours. 
If $y_u$ is still infinity, then we know that $\delta^+(u) = \emptyset$ or $y_v = \infty$ for all $v\in N^+(u)$.
In this case, we find a flow-absorbing cycle at $u$ in $G^{(k)}$ using the multiplicative Bellman--Ford algorithm, by treating the gain factors as arc costs.
If there is none, then we proceed to the next phase immediately as $y_u$ is unbounded from above in the subsystem $(G^{(k)},c,\gamma)$.
This is because $u$ cannot reach a flow-absorbing cycle in $G^{(k)}$ by induction.
We would like to point out that this does not necessarily imply that the full system $(G,c,\gamma)$ is feasible (see Appendix \ref{sec:2vpi_finite} for details).
On the other hand, if Bellman--Ford returns a flow-absorbing cycle, then $y_u$ is set to the upper bound implied by the cycle.
Then, we apply Algorithm \ref{alg:newton_fast} to solve \eqref{sys:2vpi_dual} for the subsystem $(G^{(k)},c,\gamma)$.

The value and supergradient oracle for the parametric function $f(\delta)$ is {\sc Grapevine}. 
Let $G^{(k)}_u$ be the modified graph and $\bar{y}\in \bar{\R}^{n+1}$ be the node labels as defined in the previous subsection.
In order to provide Algorithm \ref{alg:newton_fast} with a suitable initial point and supergradient, we run {\sc Grapevine} on $G^{(k)}_u$ with input node labels $\bar{y}$.
It updates $\bar{y}$ and returns a walk $P$ from $u$.
If $\bar{y}_V$ is not feasible to the dual LP for $f(\bar{y}_{u'})$ or $P$ is a non-trivial walk with $\gamma(P)\geq 1$, then 
we declare infeasibility.
Otherwise, we run Algorithm \ref{alg:newton_fast} with the initial point $\bar{y}_{u'}$ and supergradient $\gamma(P) - 1$. 
We remark that {\sc Grapevine} continues to update $\bar{y}$ throughout the execution of Algorithm \ref{alg:newton_fast}.

\begin{theorem}\label{thm:correct}
If Algorithm \ref{alg:main} returns $y\in \bar{\R}^n$, then $y=y^{\max}$ if the M2VPI system is feasible.
Otherwise, the system is infeasible.
\end{theorem}

\begin{proof}
It suffices to prove the theorem for the subsystem $(G^{(k)},c,\gamma)$ encountered in each phase $k$.
We proceed by induction on $k$.
For the base case $k=0$, the system $(G^{(0)},c,\gamma)$ is trivially feasible as it does not have any constraints.
Hence, $ y^{\max} = (\infty,\infty,\dots,\infty) = y$, where the second equality is due to our initialization.
For the inductive step, assume that the theorem is true for some $0\leq k < n$ and consider the system $(G^{(k+1)}, c,\gamma)$.
If Algorithm \ref{alg:main} terminated in phase $k$, then $(G^{(k+1)}, c,\gamma)$ is infeasible by the inductive hypothesis.
So, let $y\in \bar{\R}^n$ be the node labels maintained by the algorithm during Line \ref{line:label_correcting} of phase $k+1$.
We have $y_u = \infty$ if and only if $\mathcal{C}^{abs}_u(G^{(k+1)}) = \emptyset$ and $y_v = \infty$ for all $v\in N^+(u)$.
For each $v\neq u$, we also have $y_v = \infty$ if and only if $v$ cannot reach a flow-absorbing cycle in $G^{(k)}$.
So, if $y_u = \infty$, then $u$ cannot reach a flow-absorbing cycle in $G^{(k+1)}$.
By the inductive hypothesis, $y = y^{\max}$ if the system $(G^{(k+1)}, c,\gamma)$ is feasible.

Next, assume that $y_u<\infty$.
Without loss of generality, we may assume that every node $v$ with $y_v = \infty$ can reach $u$ in $G^{(k+1)}$.
Let $W := \set{v\in V:y_v = \infty}$.
Note that the cut $W$ does not have any outgoing edges in $G^{(k+1)}$.
If there exists a negative unit-gain cycle in $G^{(k+1)}[W]$, then it contains a violated arc with respect to any finite labels.
In this case, the algorithm correctly detects infeasibility.
Otherwise, by Lemma \ref{lem:2vpi_infty}, $f(\delta')>-\infty$ for a sufficiently high $\delta'\in \R$ because there are no flow-absorbing cycles in $G^{(k+1)}[W]$.
Pick $\delta'> y_u$ big enough such that an optimal dual solution $y'\in \R^n$ to $f(\delta')$ satisfies $y'_v = y_v$ for all $v\in V\setminus W$.
Among all such optimal dual solutions, choose $y'$ as the pointwise maximal one.
Then, every vertex $v\in W$ has a tight path to $u$ in $G^{(k+1)}$.
Now, let $\bar{y}'\in \R^{n+1}$ be node labels defined by $\bar{y}'_{u'} := y_u$ and $\bar{y}'_v := y'_v$ for all $v\in V$.
It is easy to see that running {\sc Grapevine} on $G^{(k+1)}_u$ with input node labels $\bar{y}$ and $\bar{y}'$ yield the same behaviour.
Let $(\bar{z},P)$ be the node labels and walk returned by {\sc Grapevine}.

Let $x\in \R_+^{E(G^{(k+1)})}$ be a feasible solution to \eqref{sys:2vpi_frac} such that $y_u = c^{\top}x/(1-\sum_{e\in \delta^-(u)}\gamma_e x_e)$.
Clearly, such an $x$ exists if $y_u = c(C)/1-\gamma(C)$ for some flow-absorbing cycle $C\in \mathcal{C}^{abs}_u(G^{(k+1)})$.
Otherwise, if $y_u = c_{uv} + \gamma_{uv}y_v$ for some $uv\in \delta^+(u)$, then $y_v = c(Q) + \gamma(Q)(c(C)/1-\gamma(C))$ where $Q$ is a path leading to a flow-absorbing cycle $C$ in $G^{(k)}[V\setminus W]$. 
This is because $y_{V\setminus W}$ is the pointwise maximal solution to the feasible subsystem $(G^{(k)}[V\setminus W],c,\gamma)$ by the inductive hypothesis.
Hence, $x$ can be chosen as the fundamental flow from $u$ to the cycle $C$ via the path $Q+{uv}$.

Now, according to Lemma \ref{lem:grapevine}, if $\bar{z}_V$ is not feasible to the dual LP for $f(y_u)$, then $f(y_u) = -\infty$.
By Lemma \ref{lem:2vpi_init}, the optimal value of \eqref{sys:2vpi_frac} is $-\infty$.
On the other hand, if $\bar{z}_V$ is a feasible solution to the dual LP for $f(y_u)$, then it is also optimal. 
Moreover, $P$ is a shortest path from $u$ with respect to $\bar{y}'$ in $G^{(k+1)}_u$.
If $E(P)>0$ and $\gamma(P)\geq 1$, then the path ends at $u'$ because $\bar{y}'$ is dual feasible to $f(\delta')$. 
Let $\bar{x}$ be the fundamental $u$-$u'$ flow on $P$.
By complementary slackness, $\bar{x}$ is an optimal primal solution to $f(y_u)<0$ and $1-\sum_{e\in \delta^-(u)}\gamma_e\bar{x}_e = 1-\gamma(P) \leq 0$.
Applying Lemma \ref{lem:2vpi_init} again yields unboundedness of $\eqref{sys:2vpi_frac}$.
In both cases, as \eqref{sys:2vpi_primal} is feasible, $(G^{(k+1)},c,\gamma)$ is infeasible.

If the above cases do not apply, then $\bar{z}_u$ and $\gamma(P)-1$ constitute a suitable initial point and supergradient for Algorithm \ref{alg:newton_fast} respectively.
Note that the node labels $\bar{y}$ are updated to $\bar{z}\in \R^{n+1}$.
Throughout the execution of Algorithm \ref{alg:newton_fast}, it is easy to see that $\bar{y}_V$ remain an upper bound on every feasible solution to the system $(G^{(k+1)},c,\gamma)$.
If phase $k+1$ terminates with node labels $y:=\bar{y}_V$, then $y_u$ is the largest root of $f$.
By Lemma \ref{lem:linfrac}, $y_u$ is the optimal value of \eqref{sys:2vpi_frac}.
Since $y$ is an optimal solution to \eqref{sys:2vpi_dual}, we obtain $y = y^{\max}$ as desired.
On the other hand, if phase $k+1$ terminates with {\tt INFEASIBLE}, then $f$ does not have a root.
By Lemma \ref{lem:linfrac}, the optimal value of \eqref{sys:2vpi_frac} is $-\infty$.
As \eqref{sys:2vpi_primal} is feasible, this implies that $(G^{(k+1)},c,\gamma)$ is infeasible.
\end{proof}

To bound the running time of Algorithm \ref{alg:main}, it suffices to bound the running time of Algorithm \ref{alg:newton_fast} in every phase.
Our strategy is to analyze the sequence of paths whose gain factors determine the right derivative of $f$ at each iterate of Algorithm \ref{alg:newton_fast}.
The next property is crucial our arc elimination argument.

\begin{definition}
Let $\mathcal{P} = (P^{(1)}, P^{(2)}, \dots, P^{(\ell)})$ be a sequence of paths from $u$. 
We say that $\mathcal{P}$ satisfies \emph{subpath monotonicity at $u$} if for every pair $P^{(i)},P^{(j)}$ where $i<j$ and for every shared node $v\neq u$, we have $\gamma(P^{(i)}_{uv}) \leq \gamma(P^{(j)}_{uv})$.
\end{definition}

\begin{lemma}\label{lem:subpath}
Let $\delta^{(1)}>\delta^{(2)}>\dots>\delta^{(\ell)}$ be a decreasing sequence of iterates.
For each $\delta^{(i)}\in \R$, let $P^{(i)}$ be a $u$-$u'$ path in $G_u$ on which a unit flow is an optimal primal solution to $f(\delta^{(i)})$.
Then, the sequence $(P^{(1)}, P^{(2)}, \dots, P^{(\ell)})$ satisfies subpath monotonicity at $u$.
\end{lemma}

\begin{proof}
For each $i\in [\ell]$, let $y^{(i)}\in \R^{n}$ be an optimal dual solution to $f(\delta^{(i)})$. 
Let $\bar{y}^{(i)}\in \R^{n+1}$ be the node labels in $G_u$ defined by $\bar{y}^{(i)}_{u'} := \delta^{(i)}$ and $\bar{y}_v := y_v$ for all $v\neq u'$. 
By complementary slackness, every edge in $P^{(i)}$ is tight with respect to $\bar{y}^{(i)}$.
Hence, $P^{(i)}$ is a shortest $u$-$u'$ path in $G_u$ with respect to $\bar{y}^{(i)}$.
Now, pick a pair of paths $P^{(i)}$ and $P^{(j)}$ such that $i<j$ and they share a node $v\neq u$.
Then, the subpaths $P^{(i)}_{uv}$ and $P^{(j)}_{uv}$ are also shortest $u$-$v$ paths in $G_u$ with respect to $\bar{y}^{(i)}$ and $\bar{y}^{(j)}$ respectively.
Observe that $\bar{y}^{(i)}_v > \bar{y}^{(j)}_v$ because $\bar{y}^{(i)}_{u'} = \delta^{(i)} > \delta^{(j)} = \bar{y}^{(j)}_{u'}$.
Define the function $\psi:[\bar{y}^{(j)}_v, \bar{y}^{(i)}_v]\rightarrow\bar{\R}$ as
\[\psi(\alpha) := \inf\set{c(P) + \gamma(P)\alpha: P \text{ is a $u$-$v$ walk in $G_u$}}.\] 
Clearly, it is increasing and concave.
It is also finite because $\psi(\bar{y}^{(i)}_v) = c(P^{(i)}_{uv}) + \gamma(P^{(i)}_{uv})\bar{y}^{(i)}_v$ and $\psi(\bar{y}^{(j)}_v) = c(P^{(j)}_{uv}) + \gamma(P^{(j)}_{uv})\bar{y}^{(j)}_v$.
Subpath monotonicity then follows from concavity of $\psi$.
\end{proof}

\begin{theorem}\label{thm:runtime}
In each phase $k$ of Algorithm \ref{alg:main}, Algorithm \ref{alg:newton_fast} terminates in $O(|E(G^{(k)})|)$ iterations. 
\end{theorem}

\begin{proof}
Fix a phase $k\in [n]$ and denote $m_k := |E(G^{(k)})|$.
Let $\bar{\mathcal{Y}} = (\bar{y}^{(1)}, \bar{y}^{(2)}, \dots, \bar{y}^{(\ell)})$ be the sequence of node labels at the start of every iteration of Algorithm \ref{alg:newton_fast} in phase $k$.
Note that $\bar{y}^{(i)}\geq \bar{y}^{(i+1)}$ and $\bar{y}^{(i)}_{u'} > \bar{y}^{(i+1)}_{u'}$ for all $i<\ell$.
Let $f:\R\rightarrow \bar{\R}$ be the parametric function associated with the linear fractional program \eqref{sys:2vpi_frac} for the subsystem $(G^{(k)},c,\gamma)$.
We may assume that $\ell \geq 1$, which in turn implies that $f(y^{(1)}_{u'})$ is finite by Lemma \ref{lem:2vpi_init}.
By Lemma \ref{lem:2vpi_infty}, there are no negative unit-gain cycles or bicycles in $G^{(k)}\setminus \delta^+(u)$.
It follows that all negative unit-gain cycles and negative bicycles in $G^{(k)}$ contain $u$. 
Hence, there exists a smallest $\varepsilon\geq 0$ such that the subsystem $(G^{(k)},\hat{c},\gamma)$ is feasible, where $\hat{c}\in \R^{m_k}$ are modified arc costs defined by $\hat{c}_e := c_e + \varepsilon$ if $e\in \delta^+(u)$ and $\hat{c}_e := c_e$ otherwise.

For each $i>1$, every basic optimal primal solution to $f(\bar{y}^{(i)}_{u'})$ is a path flow from $u$ to $u'$ in $G^{(k)}_u$.
This is because $u$ cannot reach a flow-absorbing cycle in the subgraph of $G^{(k)}_u$ induced by tight arcs with respect to $\bar{y}^{(i)}_u$.
Indeed, such a cycle would impose an upper bound of $\bar{y}^{(i)}_u$ on the variable $y_u$.
As $\bar{y}^{(i-1)}_u>\bar{y}^{(i)}_u$, this contradicts the feasibility of $\bar{y}^{(i-1)}_V$ to the dual LP for $f(\bar{y}^{(i-1)}_{u'})$.
For each $i>1$, let $P^{(i)}$ be a $u$-$u'$ path with the smallest gain factor in the subgraph of $G^{(k)}_u$ induced by tight arcs with respect to $\bar{y}^{(i)}$.
Note that $P^{(i)}$ is well-defined due to the same reason above.
Then, $\gamma(P^{(i)})-1 = \min \partial f(\bar{y}^{(i)}_{u'})$.
Denote this sequence of $u$-$u'$ paths as $\mathcal{P}:=(P^{(2)},P^{(3)},\dots, P^{(\ell)})$.

Without loss of generality, we may assume that $\bar{y}^{(i)}$ is finite for all $i\geq 1$.
Since every vertex can reach a flow-absorbing cycle in $G^{(k)}$, there exists a pointwise maximal solution $y^*\in \R^n$ to the modified system $(G^{(k)},\hat{c},\gamma)$.
Define the reduced cost $c^*\in \R^{m_k}_+$ as $c^*_{vw} := \hat{c}_{vw} + \gamma_{vw}y^*_w - y^*_v$ for all $vw\in E(G^{(k)})$. 
Since $f(y^*_u) = -\varepsilon$, we obtain
\begin{align*}
	c^*(P^{(i)}) &=  c(P^{(i)}) - (1-\gamma(P^{(i)}))y^*_u + \varepsilon \\
	&= f(\bar{y}^{(i)}_{u'}) - (1-\gamma(P^{(i)}))(y^*_u - \bar{y}^{(i)}_u) - f(y^*_u) \\
	&= D_f(y^*_u,\bar{y}^{(i)}_{u'}) \leq \frac12 D_f(y^*_u,\bar{y}^{(i-2)}_{u'}) = \frac12 c^*(P^{(i-2)})
\end{align*}
for all $i>3$, where the inequality is due to Lemma \ref{lem:bregman_decay}.

Consider the vector $x\in \R^{m_k}_+$ defined by 
\[x_{vw} := \begin{cases}
	\max_{i\in [\ell]} \set{\gamma(P^{(i)}_{uv}):vw\in E(P^{(i)})} &\text{ if }vw\in \cup_{i=1}^\ell E(P^{(i)}),\\
	0 &\text{ otherwise.}
\end{cases}\]
By Lemma \ref{lem:subpath}, the sequence $\mathcal{P}$ satisfies subpath monotonicity at $u$. 
Hence, $x_{vw}$ is equal to the gain factor of the $u$-$v$ subpath of the last path in $\mathcal{P}$ that contains $vw$. 
Let $0\leq c^*_1x_1\leq c^*_2x_2 \leq \dots \leq c^*_{m_k}x_{m_k}$ be the elements of $c^* \circ x$ in nondecreasing order.
Let $e_1, e_2, \dots, e_{m_k}$ denote the arcs in $G^{(k)}$ according to this order, and define $d_i := \sum_{j=1}^i c^*_jx_j$ for every $i\in [m_k]$. 
Then, $c^*(P^{(i)})\in [d_1, d_{m_k}]$ for all $i\in [\ell]$ because $c^*(P^{(\ell)}) \geq d_1$ and $c^*(P^{(1)})\leq d_{m_k}$. 
To prove that $\ell= O(m_k)$, it suffices to show that every interval $(d_i, d_{i+1}]$ contains the cost of at most two paths from $\mathcal{P}$. 

Pick $j<m_k$. 
Among all the paths in $\mathcal{P}$ whose costs lie in $(d_j, d_{j+1}]$, let $P^{(i)}$ be the most expensive one. 
If $d_j\geq d_{j+1}/2$, then 
\[c^*(P^{(i+2)}) \leq \frac{1}{2}c^*(P^{(i)}) \leq \frac{1}{2} d_{j+1} \leq d_j.\]
On the other hand, if $d_j<d_{j+1}/2$, then 
\[c^*(P^{(i+2)}) \leq \frac{1}{2}c^*(P^{(i)}) \leq \frac{1}{2}d_{j+1} = d_{j+1} - \frac{1}{2}d_{j+1} = c^*_{j+1}x_{j+1} + d_j - \frac{1}{2}d_{j+1} < c^*_{j+1}x_{j+1}.\]
By subpath monotonicity, the paths from $P^{(i+2)}$ onwards do not contain an arc from the set $\set{e_{j+1}, e_{j+2}, \dots, e_{m_k}}$. 
Therefore, their costs are at most $d_j$ each.
\end{proof}

The runtime of every iteration of Algorithm \ref{alg:newton_fast} is dominated by {\sc Grapevine}.
Thus, following the discussion in Appendix \ref{sec:2vpi_finite}, we obtain the following result.

\begin{corollary}
Algorithm \ref{alg:main} solves the feasibility of M2VPI linear systems in $O(m^2n^2)$ time.
\end{corollary} 

One might wonder if Algorithm \ref{alg:main} is still strongly polynomial if we replace the look-ahead Newton--Dinkelbach method on Line \ref{line:newton} with the standard version.
In Appendix \ref{sec:newton_analysis}, we show that this is indeed the case, though with a slower convergence.

\subsection{Deterministic Markov Decision Processes}
\label{sec:dmdp}

In this subsection, we replace \textsc{Grapevine} with a variant of Dijkstra's algorithm (Algorithm \ref{alg:dijkstra}) in order to speed up Algorithm \ref{alg:main} for solving a special class of 2VPI linear programs, known as deterministic Markov decision processes (DMDPs).
This idea was briefly mentioned by Madani in \cite{conf/aaai/Madani02}; we will supply the details.  Recall that an instance of DMDP is described by a directed multigraph $G=(V,E)$ with arc costs $c\in \R^m$ and \emph{discount} factors $\gamma\in (0,1]^m$. 
The goal is to select an outgoing arc from every node so as to minimize the total discounted cost over an infinite time horizon.
It can be formulated as the following pair of primal and dual LPs.

\begin{minipage}[t]{.48\textwidth} \vspace{-3mm}
\begin{align*}\tag{P}\label{eq:dmdp_primal}
   &\min\; c^{\top}x \\
   &\subto\;\, \nabla x_v = 1 \qquad \forall v\in V \\
   &\qquad\quad\;\; x \geq 0 
\end{align*}
\end{minipage}
\begin{minipage}[t]{.48\textwidth} \vspace{-3mm}
\begin{align*}\tag{D}\label{eq:dmdp_dual}
	 &\max\; \1^{\top}y \\
	 &\subto\;\; y_v - \gamma_e y_w \leq c_e \qquad \forall e=(v,w)\in E 
\end{align*}
\end{minipage} \vspace{3mm} 

Since the discount factor of every cycle is at most 1, there are no bicycles in $G$. 
Consequently, by Theorem \ref{thm:feas_char}, the linear program \eqref{eq:dmdp_dual} is infeasible if and only if there is a negative unit-gain cycle in $G$.
This condition can be easily checked by running a negative cycle detection algorithm on the subgraph induced by arcs with discount factor 1.

Algorithm \ref{alg:dijkstra} is slightly modified from the standard Dijkstra's algorithm \cite{journals/nm/Dijkstra59} to handle our notion of shortest paths that depends on node labels.
As part of the input, it requires a target node $t$ with out-degree zero, node labels $y\in \R^n$ which induce nonnegative reduced costs, and a parameter $\alpha<y_t$.
As output, it returns a shortest path tree $T$ to $t$ when $y_t$ is decreased to $\alpha$.
It also returns node labels $z\in \R^n$ which certify the optimality of $T$, i.e.~$z$ induces nonnegative reduced costs with zero reduced costs on $T$, and $z_t = \alpha$.

\begin{algorithm}[htb!]
	\caption{Recompute shortest paths to $t$}
	\label{alg:dijkstra}

	\Input{A directed multigraph $G=(V,E)$ with arc costs $c\in \R^E$ and discount factors $\gamma\in (0,1]^E$, a target node $t\in V$ where $\delta^+(t) = \emptyset$, node labels $y\in \R^V$ such that $c_{vw} + \gamma_{vw}y_w - y_v \geq 0$ for every $vw\in E$, and a parameter $\alpha<y_t$}
	\Output{An in-tree $T$ rooted at $t$ and node labels $z\in \R^V$ such that $z\leq y$, $z_u = \alpha$ and $c_{vw} + \gamma_{vw}z_w - z_v \geq 0$ for every $vw\in E$, with equality on every arc of $T$.}
	\BlankLine
	$y_u\leftarrow \alpha$\;
	Define reduced cost $\bar{c}\in \R^E$ by $\bar{c}_{vw}\leftarrow c_{vw} + \gamma_{vw}y_w - y_v$ for all $vw\in E$\;
	Initialize node labels $z\in \R^V$ by $z_v\leftarrow 0$ for all $v\in V$\;
	Initialize sets $R\leftarrow \set{t}$ and $S\leftarrow \emptyset$\;
	\While{$R\neq \emptyset$}{
		$w \leftarrow \argmin_{v\in R}\set{z_v}$\;
		$R \leftarrow R \setminus \set{w}$\;
		$S \leftarrow S \cup \set{w}$\;
		\ForEach{$vw\in E$ where $v\notin S$}{
			\If{$z_v > \bar{c}_{vw} + \gamma_{vw}z_w$}{
				$z_v\leftarrow \bar{c}_{vw} + \gamma_{vw}z_w$\;
				$\pred(v) \leftarrow vw$\;
				$R\leftarrow R \cup \set{v}$\;
			}
		}
	}
	Let $T$ be the in-tree defined by $\pred()$\;
	$z \leftarrow y+z$\;
	\Return $(z,T)$\;
\end{algorithm}

An \emph{iteration} of Algorithm \ref{alg:dijkstra} refers to a repetition of the while loop.
In the pseudocode, observe that $\bar{c}_e\geq 0$ for all $e\in E\setminus \delta^-(u)$. 

\begin{lemma}
Algorithm \ref{alg:dijkstra} is correct.
\end{lemma}

\begin{proof}
We proceed by induction on the number of elapsed iterations $k$. Let $z$ be the node labels at the end of iteration $k$. 
For each $i\leq k$, let $v_i$ be the node added to $S$ in iteration $i$. Note that $z_S$ remains unchanged in future iterations. We first show that $z_{v_2}\leq z_{v_3}\leq \dots \leq z_{v_k} <z_{v_1} = 0$. The base case $k=1$ is true due to our initialization, while the base case $k=2$ is true because $v_2\in R$. For the inductive step, suppose that the claim is true for some $k\geq 2$. Let $v_{k+1}=\argmin_{v\in R}\set{z_v}$ and $v_j = \pred(v_{k+1})$ for some $j\leq k$. We know that $z_{v_{k+1}} < 0$ because $v_{k+1}\in R$. If $j<k$, then $z_{v_{k+1}}\geq z_{v_k}$, as otherwise $v_k$ would not have been chosen to enter $S$ in iteration $k$. If $j=k$, using the fact that $\gamma_{v_{k+1}v_k}\leq 1$ and $\bar{c}_{v_{k+1}v_k}\geq 0$, we obtain
\[z_{v_{k+1}} = \bar{c}_{v_{k+1}v_k} + \gamma_{v_{k+1}v_k}z_{v_k} \geq z_{v_k}.\]
It is left to show that $\bar{c}_{vw} + \gamma_{vw}z_w - z_v\geq 0$ for all $vw\in E(G[S])$. The base case $k=1$ is trivially true. For the inductive step, suppose that the statement is true for some $k\geq 1$. We know that $z_{v_{k+1}}\leq \bar{c}_{v_{k+1}v} + \gamma_{v_{k+1}v}z_v$ for every outgoing arc $v_{k+1}v\in E(G[S])$. 
For every incoming arc $vv_{k+1}\in E(G[S])$, using the fact that $\gamma_{vv_{k+1}}\leq 1$ and $\bar{c}_{vv_{k+1}}\geq 0$, we get
\[z_v \leq \bar{c}_{vv_{k+1}} + \gamma_{vv_{k+1}}z_v \leq \bar{c}_{vv_{k+1}} + \gamma_{vv_{k+1}}z_{v_{k+1}},\]
where the second inequality follows from $z_v\leq z_{v_{k+1}}$.
\end{proof}

In every phase $k$ of Algorithm \ref{alg:main}, Algorithm \ref{alg:dijkstra} now replaces {\sc Grapevine} as the new value and supergradient oracle of $f$.
Given an optimal dual solution $y$ to $f(\alpha)$ for some $\alpha\in \R$, Algorithm \ref{alg:dijkstra} is used to compute an optimal dual solution to $f(\alpha')$ for any $\alpha'<\alpha$.
In particular, we run it on the modified graph $G^{(k)}_u$ with input node labels $\bar{y}$ defined by $\bar{y}_{u'} := \alpha$ and $\bar{y}_v := y_v$ for all $v\neq u'$, target node $t=u'$, and parameter $\alpha'<\alpha$.
Note that $u'$ has out-degree zero in $G^{(k)}_u$ by construction.
Let $(\bar{z},T)$ be the node labels and tree returned by Algorithm \ref{alg:dijkstra}, where $\bar{z}_V$ is an optimal dual solution to $f(\alpha')$.
A supergradient at $f(\alpha')$ can be inferred from the output via complementary slackness.
Specifically, if $u\in V(T)$, then $\gamma(P)-1\in \partial f(\alpha')$ where $P$ is the unique $u$-$u'$ path in $T$.
Otherwise, $u$ can reach a flow-absorbing cycle in the tight subgraph with respect to $\bar{z}$, so $-1\in \partial f(\alpha')$.

An efficient implementation of Dijkstra's algorithm using Fibonacci heaps was given by Fredman and Tarjan \cite{journals/jacm/FredmanT87}. 
It can also be applied to our setting, with the same running time of $O(m+n\log n)$.
Consequently, we obtain a faster running time of Algorithm \ref{alg:main} for DMDPs. 

\begin{corollary}
Algorithm \ref{alg:main} solves deterministic MDPs in $O(mn(m+n\log n))$ time.
\end{corollary}

\section{Parametric Submodular Function Minimization}
\label{sec:sfm}

Let $V$ be a set with $n$-elements and define $2^V := \{S: S \subseteq V\}$
to be the set of all subsets of $V$. A function $h: 2^{V} \rightarrow \R$ is
submodular if \[
h(S) + h(T) \geq h(S \cap T)
+ h(S \cup T)\quad \forall S,T \subseteq V\, .
\]
 Given non-negative submodular function $h: 2^{V} \rightarrow
\R_+$ and a vector $a \in \R^V$ satisfying $\max_{i \in V} a_i > 0$, we
examine the problem of computing
\begin{equation}
\delta^* := \max \{\delta: \min_{S \subseteq V} h(S) - \delta a(S) \geq 0\}, \label{def:p-sfm}
\end{equation}
where $a(S) := \sum_{i \in S} a_i$. As the input model, we assume access to an 
evaluation oracle for $h$, which allows us to query $h(S)$ for any set $S
\subseteq V$. The above problem models the line-search problem inside a
submodular polyhedron and has been studied
in~\cite{conf/ipco/GoemansGJ17,journals/do/Nagano07,journals/or/Topkis78}.

To connect to the root finding problem studied in previous sections, for $\delta \in \R$, we define
\[
f(\delta) :=  \min_{S \subseteq V} h_\delta(S) := \min_{S \subseteq V} h(S) - \delta a(S).
\] 
Since $f$ is the minimum of $2^n$ affine functions, $f$ is a piecewise linear
concave function. Noting that $f$ is continuous, problem~\eqref{def:p-sfm}
can be equivalently restated as that of computing the largest root of $f$,
i.e., the largest $\delta^* \in \R$ such that $f(\delta^*) = 0$. The
assumption that $h$ is non-negative ensures that $f(0) \geq 0$, and the
assumption that $\max_{i \in V} a_i > 0$ ensures that $\delta^*$ exists and
$\delta^* \geq 0$ (see the initialization section below).  Given the root
finding representation, we may apply the Newton--Dinkelbach method on $f$ to
compute $\delta^*$. This approach was taken by Goemans, Gupta and
Jaillet~\cite{conf/ipco/GoemansGJ17}, who were motivated to give a more efficient
alternative to the parametric search based algorithm of
Nagano~\cite{journals/do/Nagano07}. Their main result is as follows:

\begin{theorem}
\label{thm:p-sfm}
The Newton-Dinkelbach method requires at most $n^2 + O(n \log^2 n)$ iterations to solve~\eqref{def:p-sfm}.
\end{theorem}

The goal of this section is to give a simplified potential function based
proof of the above theorem using the \emph{accelerated} Newton--Dinkelbach
method (Algorithm~\ref{alg:newton_fast}), where we will give a slightly
weaker $2n^2+2n+4$ bound on the iteration count.  Our analysis uses the same combinatorial ring family analysis as in \cite{conf/ipco/GoemansGJ17}, but the Bregman divergence enables considerable simplifications.

\subsection{Implementing the Accelerated Newton--Dinkelbach} 

We explain how to implement and initialize the accelerated Newton--Dinkelbach
method in the present context. To begin, Algorithm~\ref{alg:newton_fast}
requires access to the supergradients of $f$. For $\delta \in \R$, it is easy
to verify that
\[
S \in {\rm argmin} \{h_\delta(T): T \subseteq V\} \Rightarrow -a(S) \in \partial f(\delta).
\]
Therefore, computing supergradients of $f$ can be reduced 
to computing minimizers of the submodular functions $h_\delta(S) :=
h(S) - \delta a(S)$, $\delta \in \R$. Submodular function minimization (SFM) is
a classic problem in combinatorial optimization and has been extensively
studied from the viewpoint of strongly polynomial
algorithms~\cite{conf/soda/dvz,journals/MP/Iwata08,conf/soda/SO09,conf/focs/LSW15,conf/soda/jiang21}.
The fastest strongly polynomial running time is due to Jiang~\cite{conf/soda/jiang21} 
who gave an algorithm for SFM using 
$O(n^3)$ calls to the evaluation oracle.

In what follows, we assume access to an SFM oracle, that we will call on the
submodular functions $h_\delta$, for $\delta \in \R$. Each iteration of
Algorithm~\ref{alg:newton_fast} requires two calls to a supergradient oracle,
one for the standard step and one for the look-ahead step, and hence can be
implemented using two calls to the SFM oracle. Gupta, Goemans and
Jaillet~\cite{conf/ipco/GoemansGJ17} were directly concerned with the number of
calls to an SFM oracle, which is exactly equal to the number of iterations of
standard Newton--Dinkelbach (it requires only one SFM call per iteration instead of two). As mentioned above, we will prove a $2n^2+2n+4$
bound on the iteration count for accelerated Newton--Dinkelbach, which will
recover the bound on the number of SFM calls of~\cite{conf/ipco/GoemansGJ17} up to
a factor $4$. Since accelerated Newton--Dinkelbach is always as fast as the
standard method (it goes at least as far in each iteration), the iteration bound in Theorem~\ref{thm:p-sfm} in fact applies to the accelerated method as well.  

We now explain how to initialize the method. For this purpose,
Algorithm~\ref{alg:newton_fast} requires $\delta^{(1)} \in \R$ and $g^{(1)}
\in \partial f(\delta^{(1)})$ such that $f(\delta^{(1)})\le 0$ and $g^{(1)} < 0$. We proceed as in~\cite{conf/ipco/GoemansGJ17} and let $\delta^{(1)} :=
\argmin \{h(\{i\})/a_i: i \in V, a_i > 0\} \geq 0$, which is well-defined by
assumption on $a$. We compute $f(\delta^{(1)})$ by the SFM oracle. Note that
\[
f(\delta^{(1)}) = \min_{S \subseteq V} h_\delta(S) \leq \min_{i \in V, a_i > 0} h(\{i\})-\delta^{(1)} a_i = 0. 
\]
If $f(\delta^{(1)}) = 0$, we return $\delta^{(1)}$, as we are already done.
Otherwise if $f(\delta^{(1)}) < 0$, 
set $g^{(1)} =- a(S^{(1)})$, where $S^{(1)} \in {\rm argmin}_{S \subseteq V} h_{\delta^{(1)}}(S)$ as returned by the oracle.
From here, note that
\[
0 > f(\delta^{(1)}) = h_{\delta^{(1)}}(S^{(1)}) = h(S^{(1)}) - \delta^{(1)} a(S^{(1)}) = h(S^{(1)}) + g^{(1)} \delta^{(1)} \geq g^{(1)} \delta^{(1)},
\]
where the last inequality follows by non-negativity of $h$. Since $\delta^{(1)} \geq 0$, the above implies that $\delta^{(1)} > 0$ and $g^{(1)} < 0$. 
We may
 therefore
 initialize Algorithm~\ref{alg:newton_fast} with
$\delta^{(1)}$ and $g^{(1)}$. %

Assuming $f(\delta^{(1)}) < 0$, the largest root $\delta^*$ of
$f$ is guaranteed to exists in the interval $[0,\delta^{(1)})$. This follows
since $f$ is continuous, $f(0) = \min_{S \subseteq V} h(S) \geq 0$ (by
non-negativity of $h$) and $f(\delta^{(1)}) < 0$. In particular,
Algorithm~\ref{alg:newton_fast} on input $f,\delta^{(1)},g^{(1)}$ is
guaranteed to output the desired largest root $\delta^*$ in a finite number
of iterations (recalling that $f$ is piecewise affine with $2^n$ pieces). In
the next subsection, we prove a $2n^2 + 2n + 4$ bound on the number of
iterations.   

\subsection{Proof of the $2n^2+2n+4$ Iteration Bound}

Let $\delta^{(1)} > \cdots > \delta^{(\ell)} = \delta^{*}$ denote iterates of
Algorithm~\ref{alg:newton_fast} on input $f$ and $\delta^{(1)},g^{(1)} < 0$
as above. For each $i \in [\ell]$, let $S^{(i)}$ be an any set satisfying 
\[
S^{(i)} \in {\rm argmax} \{ a(S): S \in {\rm argmin}_{T \subseteq V} h_{\delta^{(i)}}(T) \}.
\]
It is not hard to verify that $S^{(i)}$, $i \in [\ell]$, is a minimizer of
$h_{\delta^{(i)}}$ inducing the right derivative of $f$ at $\delta^{(i)}$.
Precisely, $-a(S^{(i)}) = \inf_{g \in \partial f(\delta^{(i)})} g$, $\forall i \in
[\ell]$. We note that the sets $S^{(i)}$, $i \in [\ell]$, need not be the sets
outputted by the SFM oracle, and are only required for the analysis of the
algorithm.

Our goal is to prove that $\ell \leq 2n^2 + 2n + 4$. For this
purpose, we rely on the key idea of~\cite{conf/ipco/GoemansGJ17}, which is to
extract an increasing sequence of \emph{ring-families} from the sets $S^{(i)}$, $i \in [\ell]$.

A \emph{ring family} $\mathcal{R} \subseteq 2^V$ is a subsystem of sets that is
closed under unions and intersections, precisely $A,B \in \mathcal{R}
\Rightarrow$ $A \cap B, A \cup B \in \mathcal{R}$. Given $\mathcal{T}
\subseteq 2^V$, we let $\mathcal{R}(\mathcal{T})$ denote the smallest
ring-family containing $\mathcal{T}$. We will use the following lemma
of~\cite{conf/ipco/GoemansGJ17} which bounds the length of an increasing sequence of
ring-families:

\begin{lemma}[{\cite[Theorem 2]{conf/ipco/GoemansGJ17}}]
\label{lem:inc-seq-bnd}
Let $\emptyset \neq \mathcal{R}_1 \subsetneq \mathcal{R}_2 \subsetneq \cdots \subsetneq \mathcal{R}_k \subseteq 2^V$, where $|V| = n$. Then $k \leq \binom{n+1}{2} + 1$.
\end{lemma}

The proof of the above lemma is based on the Birkhoff representation of a
ring family. Precisely, for any ring-family $\mathcal{R} \subseteq 2^V$, with
$\emptyset,V \in \mathcal{R}$, there exists a directed graph $G$ on $V$, such
that the sets $S \in \mathcal{R}$ are exactly the subsets of vertices of $G$
having no out-neighbors. The main idea for the bound is that the digraph
representation of $\mathcal{R}_i, i \in [k],$ must lose edges as $i$
increases.  The next statement is a slightly adapted version of \cite[Theorem 5]{conf/ipco/GoemansGJ17} that it sufficient for our purposes. It shows that a sequence of sets with geometrically increasing $h$ values forms an increasing sequence of ring families. We include a proof for completeness.

\begin{lemma}\label{lem:ring-family-increase}
Let $h : 2^V \to \mathbb{R}_+$ be a non-negative submodular function.
Consider a sequence of distinct sets $T_1,T_2,\ldots,T_q \subseteq V$ such
that $h(T_{i+1}) > 4h(T_i)$ for $i \in [q-1]$. Then $T_{i+1} \notin
\mathcal{R}(\{T_1,\ldots ,T_i\})$ for all $i \in [q-1]$. 
\end{lemma}
\begin{proof}
Let $\mathcal{R}_i := \mathcal{R}(\{T_1,\dots,T_i\})$, $\forall i \in [q]$. 
We claim that $\max_{S \in \mathcal{R}_i} h(S) \leq 2h(T_i)$, $\forall i
\in [q]$. This proves $h(T_{i+1}) \notin \mathcal{R}_{i}$, for $i \in [q-1]$, 
since $h(T_{i+1}) > 4h(T_i) \geq 2h(T_i) \geq \max_{S \in \mathcal{R}_{i}}
h(S)$, noting that the second inequality uses that $h$ is non-negative. 

We now prove the claim by induction on $i \in [q]$. The base case $i=1$ is
trivial since $\mathcal{R}_1 = \{T_1\}$. We now assume that $\max_{S \in
\mathcal{R}_i} h(S) \leq 2 h(T_i)$, for $1 \leq i \leq q-1$, and prove the
corresponding bound for $\mathcal{R}_{i+1}$. Recalling that
$\mathcal{R}_{i+1}$ is the ring-family generated by $\mathcal{R}_i$ and
$T_{i+1}$, it is easy to verify that the set system 
\[
\mathcal{R}_i \cup \{T_{i+1}\} \cup \{S \cup T_{i+1}: S \in \mathcal{R}_i\} \cup \{S \cap T_{i+1}: S \in \mathcal{R}_i\} \cup \{S_1 \cup (S_2 \cap T_{i+1}): S_1,S_2 \in \mathcal{R}_i\}
\]
is a ring-family and hence is equal to $\mathcal{R}_{i+1}$. It therefore
suffices to upper bound $h(X)$ for a set $X$ of the above type. For $X \in
\mathcal{R}_i$ or $X = T_{i+1}$, the bound is by assumption. For $X = S_1
\cup (S_2 \cap T_{i+1})$, $S_1,S_2 \in R_{i+1}$, we prove the bound as
follows:
\begin{align*}
h(S_1 \cup (S_2 \cap T_{i+1})) &\leq h(S_1) + h(S_2 \cap T_{i+1}) - h(S_1 \cap S_2 \cap T_{i+1}) \quad \left(\text{ by submodularity of $h$ }\right) \\
&\leq h(S_1) + h(S_2) + h(T_{i+1}) - h(S_1 \cup T_{i+1}) - h(S_1 \cap S_2 \cap T_{i+1}) \\
&\leq h(S_1) + h(S_2) + h(T_{i+1}) \quad \left(\text{ by non-negativity of $h_{\delta^*}$ }\right) \\
&\leq 4 h(T_i) + h(T_{i+1}) \quad \left(\text{ by the induction hypothesis }\right) \\
&\leq 2 h(T_{i+1}). \quad \left(\text{ since } 4h(T_{i}) < h(T_{i+1})~\right) 
\end{align*}
For $X = S \cup T_{i+1}$ or $X = S \cap T_{i+1}$, $S \in \mathcal{R}_i$, similarly to the above, one has 
\[
h(X) \leq h(S) + h(T_{i+1}) \leq 2 h(T_i) + h(T_{i+1}) \leq \frac{3}{2} h(T_{i+1}), \text{ as needed.}
\]
\end{proof}

We now use the Bregman-divergence analysis to show that for the function $h_{\delta^*}$, the sequence of sets $T_i=S^{(\ell-4(j-1))}$, $1 \leq i \leq \lfloor \frac{\ell + 3}{4} \rfloor$ satisifes the conditions of this lemma. %
Combined with~\Cref{lem:inc-seq-bnd}, we get that the number of iterations satisfies
\[
\lfloor (\ell + 3)/4\rfloor \leq \binom{n+1}{2}+1 \Rightarrow \ell \leq 2n^2 + 2n +
4, \text{ as needed.}
\]
\begin{lemma} Let us define 
\[
T_i := S^{(\ell-4(i-1))}\, ,\quad  i\in [q]\quad\mbox{for }q:=\left\lfloor \frac{\ell + 3}{4}\right\rfloor\, .
\]
Then, the function $h_{\delta^*}$ and the sequence of sets $T_1,T_2,\ldots,T_q$ satisfy the conditions in \Cref{lem:ring-family-increase}.
\end{lemma}
\begin{proof}
The function $h_{\delta^*}$ is clearly submodular, and its minimum is 0 since
$0 = f(\delta^*) = \min_{S \subseteq V} h_{\delta^*}(S) =
h_{\delta^*}(S^{(\ell)})=h_{\delta^*}(T_1)$. In particular, $h_{\delta^*}$ is
non-negative. It is left to show $h_{\delta^*}(T_{i+1}) > 4h_{\delta^*}(T_{i})$ for $i \in [q-1]$. For each $\delta^{(i)}$, $i \in [\ell]$, we see that 
\begin{align*}
D_f(\delta^*, \delta^{(i)}) &= f(\delta^{(i)}) + \sup_{g \in \partial f(\delta^{(i)})} g(\delta^*-\delta^{(i)}) - f(\delta^*) \\  
                            &= h_{\delta^{(i)}}(S^{(i)}) - a(S^{(i)})(\delta^*-\delta^{(i)}) \quad \left(\text{ by our choice of $S^{(i)}$ and $f(\delta^*) = 0$ }\right) \\
      &= h(S^{(i)}) - \delta^{(i)} a(S^{(i)}) - a(S^{(i)})(\delta^*-\delta^{(i)}) = h_{\delta^*}(S^{(i)}).
\end{align*}
By Lemma~\ref{lem:bregman_decay} and the above, we get for $3 \leq i \leq l$ that
\begin{equation}
D_f(\delta^*, \delta^{(i)}) < \frac 12 D_f(\delta^*, \delta^{(i-2)})
\Leftrightarrow h_{\delta^*}(S^{(i)}) < \frac 12 h_{\delta^*}(S^{(i-2)}). \label{eq:sfm-exp-dec}
\end{equation}
Then, $h_{\delta^*}(T_{i+1}) > 4h_{\delta^*}(T_i)$ for $i \in [q-1]$ follows
by the definition of the $T_i$ sets.
\end{proof}

\bibliography{references}

\appendix

\begin{toappendix}
\section{Further Explanations}

\subsection{Reducing 2VPI to M2VPI }
\label{subsec: 2VPI_reduction}

Following \cite{journals/tcs/EdelsbrunnerRW89,journals/mp/HochbaumMNT93}, the idea is to replace each variable $y_u$ with $(y^+_u - y^-_u)/2$, where $y^+_u$ and $y^-_u$ are newly introduced variables.
Then, an inequality $ay_u + by_v \leq c$ becomes 
\[a\pr{\frac{y^+_u - y^-_u}{2}} + b\pr{\frac{y^+_v - y^-_v}{2}} \leq c,\]
which contains four variable, but will be adjusted based on the signs of $a$ and $b$: 
If $a$ or $b$ is zero, then the resulting inequality is already monotone and contains two variables. 
Next, if $\sgn(a) = \sgn(b)$, then we replace the inequality with $ay^+_u - by^-_v \leq c$ and $-ay^-_u + by^+_v \leq c$.
Otherwise, we replace it with $ay^+_u + by^+_v \leq c$ and $-ay^-_u - by^-_v \leq c$.
Observe that every inequality in the new system is monotone and supported on exactly two variables.
If $\hat{y}$ is a feasible solution to the original system, then setting $y^+ = \hat{y}$ and $y^- = -\hat{y}$ yields a feasible solution to the new system.
Conversely, if $(\hat{y}^+, \hat{y}^-)$ is a feasible solution to the new system, then setting $y = (\hat{y}^+ - \hat{y}^-)/2$ yields a feasible solution to the original system.
It follows that the two systems are equivalent.

\subsection{Non-existence of shortest paths}
\label{subsec:existence_shortest_cycles_or_paths}

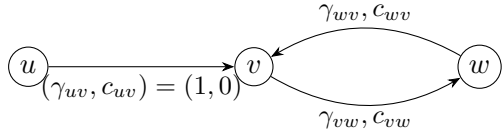
\begin{figure}[ht]
\centering
\begin{tikzpicture}[node distance=3 cm, inner sep=2.5pt, minimum size=2.5pt, auto]
	\node [vertex] (1) {$u$};
	\node [vertex] (2) [right of =1] {$v$};
	\node [vertex] (3) [right of=2] {$w$};
	
	\begin{scope}[>={Stealth[black]}]
	  	\path[->] (1) edge node[weight,below] {$(\gamma_{uv},c_{uv}) = (1, 0)$} (2);
	  	\path[->] (2) edge[bend right] node[weight,below] {$\gamma_{vw}, c_{vw}$} (3);
	  	\path[->] (3) edge[bend right] node[weight,above] {$\gamma_{wv}, c_{wv}$} (2);
  	\end{scope}
\end{tikzpicture}
\caption{A shortest path from $u$ with respect to node labels $y$ may not exist.}
\label{fig: no_shortest_path_instance}
\end{figure}

Consider \Cref{fig: no_shortest_path_instance}. We will sketch three different scenarios in which a shortest path from $u$ with respect to node labels $y\in \R^3$
does not exist. Throughout, let $C$ be the unique directed cycle %
and $C^k$ be the $v$-$v$ walk that traverses $C$ exactly $k \in \N$ times.

\paragraph{Negative unit gain cycle}
Let $\gamma_{wv} = \gamma_{vw} = 1$ and $c_{wv} = c_{vw} = -1$. Then the cycle $C$ fulfils $\gamma(C) = 1$ and $c(C) = -2 < 0$. The concatenation of $(u,v)$ and $C^k$ leads to arbitrarily short walks from $u$. In particular, there exists no shortest path from $u$. 
This observation is independent of the node labels $y$. Recall as well, that the existence of such a cycle renders the M2VPI instance infeasible (\Cref{thm:feas_char}).

\paragraph{Flow-absorbing cycle for large node labels}
Let $\gamma_{vw} = 1$ and $\gamma_{wv} = 1/2$. Then $\gamma(C) = \gamma_{vw}\gamma_{wv} = 1/2$, so $C$ is flow-absorbing. Let further $c_{wv} = c_{vw} = 0$ and $y_w = y_v = 1$. Label-correcting for the cycle $C$ then updates $y_v$ and $y_w$ in two strictly decreasing sequences, which both converge towards 0. Again, the concatenation of $(u,v)$ and $C^k$ leads to a sequence of $u$-$v$ walks that have no smallest element.

\paragraph{Flow-generating cycle for small node labels}
Let $\gamma_{vw} = 1$ and $\gamma_{wv} = 2$. Then $\gamma(C) = \gamma_{vw}\gamma_{wv} = 2$, so $C$ is flow-generating. Let further $c_{wv} = -1, c_{vw} = 0$ and $y_w = y_v = 0$. Label-correcting for the cycle $C$ then updates $y_v$ and $y_w$ in two strictly decreasing  and unbounded sequences. Again, the concatenation of $(u,v)$ and $C^k$ leads to a sequence of $u$-$v$ walks that have no smallest element.

\subsection{From $y^{\max}$ to a finite feasible solution}
\label{sec:2vpi_finite}

In this section, we show how to convert the node labels $y\in \bar{\R}^n$ obtained from Algorithm \ref{alg:main} into a finite feasible solution or an infeasibility certificate of the M2VPI system $(G,c,\gamma)$ in question.
We summarize the classical arguments already used by Aspvall and Shiloach \cite{journals/siamcomp/AspvallS80}.
If $y$ is finite, then we are done because there are no violated arcs in $G$ with respect to $y$.
In fact, $y$ is the pointwise maximal solution by Theorem \ref{thm:correct}.
So, we may assume that $y_u = \infty$ for some $u\in V$.

Define $y^{\min}\in \bar{\R}^n$ as the pointwise minimal solution to $(G,c,\gamma)$ if the system is feasible, where $y^{\min}_v := -\infty$ if and only if the variable $y_v$ is unbounded from below.
Consider the reversed graph $\ole{G}=(V,\ole{E})$, where $\ole{E}:= \set{vu:uv\in E}$ denotes the set of reversed arcs.
The cost and gain factor of each arc $e\in \ole{E}$ are given by $\ole{c}_e:= -c_{\ole{e}}/\gamma_{\ole{e}}$ and $\ole{\gamma}_e:=1/\gamma_{\ole{e}}$ respectively.
The M2VPI system defined by $(\ole{G},\ole{c},\ole{\gamma})$ is equivalent to the original system $(G,c,\gamma)$, which can be verified by performing the change of variables $z = -y$.
Let us run Algorithm \ref{alg:main} on $(\ole{G},\ole{c},\ole{\gamma})$.
By Theorem \ref{thm:correct}, if it returns node labels $z\in \bar{\R}^n$, then $z = -y^{\min}$ if the system is feasible.
Otherwise, the system is infeasible.
If $z$ is finite, then we are again done because there are no violated arcs in $\ole{G}$ with respect to $z$.
So, we may assume that $z_v = \infty$ for some $v\in V$.

If $y_w = z_w = \infty$ for some $w\in V$, then we know that $w$ cannot reach a flow-absorbing cycle in $G$ and $\ole{G}$. 
The inability to reach a flow-absorbing cycle in $\ole{G}$ is equivalent to the inability to be reached by a flow-generating cycle in $G$. 
Denote $W:=\set{w\in V:y_w = z_w = \infty}$.
Observe that every node $w\in W$ is not strongly connected to any $v\notin W$ in $G$.
Thus, checking the feasibility of the system amounts to checking whether there exists a negative unit-gain cycle in $G[W]$.
This can be done by running \textsc{Grapevine} on $G[W]$.
Let $C_1, C_2,\dots, C_k$ be the sink components in the strongly connected component decomposition of $G[W]$, and pick any $v_i\in V(C_i)$ for all $i\in [k]$.
Then, the input node labels $y'\in \bar{\R}^W$ to {\sc Grapevine} are set as $y'_{v_i}\in \R$ for all $i\in [k]$ and $y'_v := \infty$ for all other nodes. 
Let $z'\in \R^W$ be the returned node labels.
It is easy to see that there exists a negative unit-gain cycle in $G[W]$ if and only if there exists a violated arc in $G[W]$ with respect to $z'$. 

If the check above reveals that the system is feasible, then we have $y = y^{\max}$ and $-z = y^{\min}$ by Theorem \ref{thm:correct}.
Then, we can apply a result of Aspvall and Shiloach which states that the interval $[y^{\min}_u, y^{\max}_u]$ is the projection of the feasible region onto the coordinate $y_u$ for every $u\in V$.
To obtain a feasible solution, we simply fix a coordinate $y_u\in [y^{\min}_u, y^{\max}_u]$, update $y^{\min}$ and $y^{\max}$ using a generic label-correcting algorithm like {\sc Grapevine}, and repeat.

\section{2VPI Analysis without Acceleration}
\label{sec:newton_analysis}

In this section, we analyze the convergence of Algorithm \ref{alg:main} when the look-ahead Newton--Dinkelbach method is replaced with the standard version.
Interestingly, we also obtain a strongly polynomial runtime in this case, albeit slower than the accelerated version by a factor of $O(\log n)$. 
To achieve the desired runtime, we slightly strengthen Lemma \ref{lem:goemans_m}, whose proof remains largely the same.

\begin{lemma}\label{lem:goemans_n}
Let $c\in \R^m_+$ and $x^{(1)},x^{(2)},\dots,x^{(k)}\in \Z^m$ such that $\norm{x^{(i)}}_1 \leq n$ for all $i\in [k]$. If 
\[0 < c^{\top}x^{(i+1)} \leq \frac{1}{2}c^{\top}x^{(i)}\]
for all $i<k$, then $k = O(m\log n)$.
\end{lemma}

\begin{proof}[Proof of Lemma \ref{lem:goemans_n}]
Consider the polyhedron $P\subseteq \R^m$ defined by the following constraints:
\begin{align*}
	(x^{(i)} - 2x^{(i+1)})^{\top}z &\geq 0 \qquad \forall i<k \\
	(x^{(k)})^{\top}z &= 1 \\
	z &\geq 0.
\end{align*}
Let $A\in \R^{(k+m)\times m}$ and $b\in \R^{k+m}$ denote the coefficient matrix and right-hand side vector of this system. The polyhedron $P$ is nonempty because it contains the vector $c/(x^{(k)})^{\top}c$. Moreover, since $P$ does not contain a line, it has an extreme point. So there exists a vector $c'\in P$ such that $A'c'=b'$ for some nonsingular submatrix $A'\in \R^{m\times m}$ of the matrix $A$ and a subvector $b'\in \R^m$ of the vector $b$. Cramer's rule says that for each $i\in [m]$,
\[c'_i = \frac{\det A'_i}{\det A'}\]
where the matrix $A'_i$ is obtained from matrix $A'$ by replacing the $i$-th column with vector $b'$. The 1-norm of the rows of $A'_i$ is bounded by $3n$ and so by Hadamard's inequality $|\det(A_i')| \le (3n)^m$.

As the matrix $A'$ is nonsingular, we also have $\abs{\det A'}\geq 1$, which implies that $c'_i \leq (3n)^m$ for all $i\in [m]$. Finally, using the constraints which define the polyhedron $P$, we obtain
\[1 = (x^{(k)})^{\top}c'\leq \frac{(x^{(1)})^{\top}c'}{2^{k-1}}	 \leq \frac{n(3n)^m}{2^{k-1}}.\]
So, $k \leq \log(3^m n^{m+1}) + 1 = O(m \log n)$ as desired.
\end{proof}

Fix a phase $k\in [n]$ and denote $m_k = |E(G^{(k)})|$.
It is helpful to classify the iterations of the Newton--Dinkelbach method based on the magnitude by which the supergradient changes.
Recall that the supergradient at the start of iteration $i>1$ is given by $\gamma(P^{(i)})-1$, where $P^{(i)}$ is the $u$-$u'$ path returned by {\sc Grapevine} in the previous iteration.

\begin{definition}
For every $i> 1$, we say that iteration $i$ is \emph{good} if $1-\gamma(P^{(i)})\leq \frac12(1-\gamma(P^{(i-1)}))$.
Otherwise, we say that it is \emph{bad}.
\end{definition}

The next lemma gives a strongly polynomial bound on the number of good iterations.

\begin{lemma}\label{lem:good}
In each phase $k\in [n]$, the number of good iterations is $O(m_k\log k)$.
\end{lemma}

\begin{proof}
Let $\mathcal{P}$ be a sequence of $u$-$u'$ paths in $G^{(k)}_u$  at the start of every iteration of the Newton--Dinkelbach method. 
Let $\mathcal{P}^* =(P^{(1)}, P^{(2)}, \dots, P^{(t)})$ be the subsequence of $\mathcal{P}$ restricted to good iterations. 
We claim that $\gamma(P^{(i+1)})\geq \sqrt{\gamma(P^{(i)})}$ for all $i<t$. 
We use the simple inequality that $(1-x)/2\leq 1-\sqrt{x}$ for all $x\in \R_+$; one can derive this by rearranging $(\sqrt{x}-1)^2/2\ge 0$.
 This gives
\[1-\gamma(P^{(i+1)}) \leq \frac{1}{2}\pr{1-\gamma(P^{(i)})} \leq 1 -\sqrt{\gamma(P^{(i)})},\]
which proves the claim.
 Next, enumerate the arcs of each path by $P^{(i)} = (e^{(i)}_1, e^{(i)}_2, \dots, e^{(i)}_{\ell_i})$. By taking logarithms, the claim can be equivalently stated as
\[\sum_{j=1}^{\ell_{i+1}} \log\gamma_{e^{(i+1)}_j} \geq \frac{1}{2}\sum_{j=1}^{\ell_i} \log\gamma_{e^{(i)}_j}.\]
Note that both sides of the expression above are negative because $\gamma(P^{(i)})< 1$ for all $i\in [t]$. 
Let $c\in \R^{m_k}_+$ be the vector defined by $c_e = \abs{\log\gamma_e}$ for all $e\in E(G^{(k)})$. In addition, for every $i\in[t]$, define the vector $x^{(i)}\in \Z^m$ as 
\[x^{(i)}_e = -\sgn(\log \gamma_e)
	\size{\set{j\in [\ell_i]: e^{(i)}_j = e}}.\]
Then, we obtain
\[0 < c^{\top}x^{(i+1)} = \sum_{j=1}^{\ell_{i+1}} -\log\gamma_{e^{(i+1)}_j} \leq \frac{1}{2}\sum_{j=1}^{\ell_i} -\log\gamma_{e^{(i)}_j} = \frac{1}{2}c^{\top}x^{(i)}.\]
for all $i<t$. Since $\|x^{(i)}\|_1\leq k$ for all $i\in [t]$, we conclude that $t=O(m_k\log k)$ by Lemma \ref{lem:goemans_n}. 
\end{proof}

It is left to bound the number of bad iterations. 
We approach this by arguing that in a strongly polynomial number of bad iterations, an arc will no longer appear in future paths produced by the Newton--Dinkelbach method. 

\begin{lemma}\label{lem:bad}
In each phase $k\in [n]$, the number of bad iterations is $O(m_k\log k)$.
\end{lemma}

\begin{proof}
Let $\bar{\mathcal{Y}} = (\bar{y}^{(1)}, \bar{y}^{(2)}, \dots, \bar{y}^{(\ell)})$ and $\mathcal{P} = (P^{(1)}, P^{(2)}, \dots, P^{(\ell)})$ be a sequence of node labels and $u$-$u'$ paths in $G^{(k)}_u$ respectively at the start of every iteration of the Newton--Dinkelbach method. 
Without loss of generality, we may assume that $\bar{y}^{(i)}$ is finite for all $i\in [\ell]$.
For each $i\in [\ell]$, define $y^{(i)}\in \R^n$ as $y^{(i)}_u := \bar{y}^{(i)}_{u'}$ and $y^{(i)}_v := \bar{y}^{(i)}_v$ for all $v\notin \set{u,u'}$.
Now, pick an iteration $j\in [\ell]$ such that more than $\log(2n)$ bad iterations have elapsed.
Consider the reduced cost $c'\in \R^{m_k}$ given by $c'_{vw} := c_{vw} + \gamma_{vw}y^{(j)}_w - y^{(j)}_v$ for all $vw\in E(G^{(k)})$. 
Note that $c'_{vw}\geq 0$ for all $v\neq u$. 

According to Lemma \ref{lem:grapevine}, each $P^{(i)}$ is a shortest $u$-$u'$ path with respect to $\bar{y}^{(i)}$.
By complementary slackness, the unit flow on $P^{(i)}$ is an optimal primal solution to $f(\bar{y}^{(i)}_{u'})$.
Since $\bar{y}^{(i)}_{u'}>\bar{y}^{(i+1)}_{u'}$ for all $i<\ell$, the sequence $\mathcal{P}$ satisfies subpath monotonicity at $u$ by Lemma \ref{lem:subpath}.
Define the vector $x\in \R^m_{+}$ as
\[x_{vw} := \begin{cases}
	\max_{i\in [\ell]} \set{\gamma(P^{(i)}_{uv}):vw\in E(P^{(i)})} &\text{ if }vw\in \cup_{i=1}^\ell E(P^{(i)}),\\
	0 &\text{ otherwise.}
\end{cases}\]
Observe that $x_{vw}$ is the gain factor of the $u$-$v$ subpath of the last path in $\mathcal{P}$ which contains $vw$, due to subpath monotonicity. 

\begin{claim}\label{cl:h-bound}
We have $-f(\bar{y}^{(j)}_{u'}) < \norm{c'\circ x}_\infty$.
\end{claim}

\begin{proof}
For every $i\in [\ell]$, we have
\[f(\bar{y}^{(i)}_{u'}) = c(P^{(i)}) - \bar{y}^{(i)}_{u'}(1-\gamma(P^{(i)})) = c'(P^{(i)}) - (\bar{y}^{(i)}_{u'} - \bar{y}^{(j)}_{u'})(1-\gamma(P^{(i)})).\]
By applying the definition of $\bar{y}^{(i)}_{u'}$, we can upper bound its negation by
\[-f(\bar{y}^{(i)}_{u'}) = -c'(P^{(i)}) + \frac{1-\gamma(P^{(i)})}{1-\gamma(P^{(i-1)})} c'(P^{(i-1)}) \leq \abs{c'(P^{(i)})} + \abs{c'(P^{(i-1)})} \leq 2k\norm{c'\circ x}_\infty.\]
Lemma \ref{lem:monotone} tells us that $-f(\bar{y}^{(i)}_{u'})$ is nonnegative and monotonically decreasing.
Moreover, it decreases geometrically by a factor of $1/2$ during bad iterations. 
Hence, by our choice of $j$, we obtain
\[-f(\bar{y}^{(j)}_{u'}) < \left(\frac12\right)^{\log(2n)} \cdot 2k\norm{c'\circ x}_\infty = \norm{c'\circ x}_\infty. \qedhere\] 
\end{proof}

Let $d\in \R^{m_k}$ be the arc costs defined by
\[d_{vw} = \begin{cases}
	c'_{vw} &\text{ if }v\neq u,\\
	c'_{vw} - f(\bar{y}^{(j)}_{u'}) &\text{ if }v=u.
\end{cases}\]
Since $f(\bar{y}^{(j)}_{u'}) = \bar{y}^{(j)}_u - \bar{y}^{(j)}_{u'}$, observe that $d\geq 0$ because $\bar{y}^{(j)}_V$ is feasible to the dual LP for $f(\bar{y}^{(j)}_{u'})$.

\begin{claim}\label{clm:norm}
We have $\norm{d\circ x}_\infty \geq \norm{c'\circ x}_\infty$.
\end{claim}

\begin{proof}
Let $e^*=\argmax_{e\in E(G^{(k)})}\abs{c'_ex_e}$. The claim is trivial unless $e^*\in \cup_{i=1}^\ell E(P^{(i)})$ and the tail of $e^*$ is $u$.
Since $f(\bar{y}^{(j)}_{u'})\leq 0$ and $d_{e^*} = c'_{e^*} - f(\bar{y}^{(j)}_{u'})$, it suffices to show that $c'_{e^*}\geq 0$. 
For the purpose of contradiction, suppose that $c'_{e^*}< 0$. 
Since $d_{e^*} \geq 0$, this implies that $|c'_{e^*}| \leq  -f(\bar{y}^{(j)}_{u'}) < \norm{c'\circ x}_\infty$ using \Cref{cl:h-bound}. 
By the definition of $x$, $x_{e^*}=1$ because $e^*$ is the first arc of any path in $\mathcal{P}$ which uses it. 
However, this implies that
\[\abs{c'_{e^*}} = \abs{c'_{e^*}x_{e^*}} = \norm{c'\circ x}_\infty,\]
which is a contradiction.
\end{proof}

Consider the arc $e^* := \argmax_{e\in E} \abs{d_ex_e}$. 
We claim that $e^*$ does not appear in subsequent paths in $\mathcal{P}$ after iteration $j$. 
For the purpose of contradiction, suppose that there exists an iteration $i>j$ such that $e^*\in E(P^{(i)})$. 
Pick the iteration $i$ such that $P^{(i)}$ is the last path in $\mathcal{P}$ which contains $e^*$.
Since the iterates $\bar{y}_{u'}^{(\cdot)}$ are monotonically decreasing, we have
\[0 > \bar{y}^{(i+1)}_{u'} - \bar{y}^{(j)}_{u'} = \frac{c(P^{(i)})}{1-\gamma(P^{(i)})} - \bar{y}^{(j)}_{u'} = \frac{c'(P^{(i)})}{1-\gamma(P^{(i)})} = \frac{d(P^{(i)}) - f(\bar{y}^{(j)}_{u'})}{1-\gamma(P^{(i)})}\, \]
This implies that $d(P^{(i)}) < f(\bar{y}^{(j)}_{u'}) < \norm{c'\circ x}_\infty$. 
However, it contradicts
\[d(P^{(i)}) \geq  d_{e^*}x_{e^*} = \norm{d\circ x}_{\infty} \geq \norm{c'\circ x}_\infty,\]
where the first inequality is due to our choice of $i$ and the nonnegativity of $d$, while the second inequality is due to Claim \ref{clm:norm}.
Repeating the argument above for $m$ times yields the desired bound on the number of bad iterations.
\end{proof}

The runtime of every iteration of the Newton--Dinkelbach method is dominated by {\sc Grapevine}.
Thus, following the discussion in Appendix \ref{sec:2vpi_finite}, we obtain the following result. 

\begin{corollary}
If we replace Algorithm \ref{alg:newton_fast} with the Newton--Dinkelbach method in Algorithm \ref{alg:main}, then it solves the feasibility of M2VPI linear systems in $O(m^2n^2\log n)$ time.
\end{corollary}

\end{toappendix}

\end{document}